\documentclass[lettersize,journal]{IEEEtran}
\usepackage{algorithm}
\usepackage[caption=false,font=normalsize,labelfont=sf,textfont=sf]{subfig}
\usepackage{stfloats}
\usepackage{url}
\usepackage{verbatim}
\usepackage{cite}
\usepackage{amsmath,amssymb,amsfonts,bm,amsthm}
\usepackage{amsthm,epsfig}
\usepackage{bm}  
\usepackage{dsfont}
\usepackage{algorithm}
\usepackage{algpseudocode}
\usepackage{graphicx}
\usepackage{textcomp}
\usepackage{stmaryrd}
\usepackage{mathtools}
\usepackage{array,multirow}
\usepackage{xcolor}
\def\BibTeX{{\rm B\kern-.05em{\sc i\kern-.025em b}\kern-.08em
    T\kern-.1667em\lower.7ex\hbox{E}\kern-.125emX}}

\newcommand{\bxi}{\bm{\xi}}
\newcommand{\bpsi}{\bm{\psi}}
\newcommand{\bmu}{\bm{\mu}}

\newcommand{\btau}{\bm{\tau}}

\newcommand{\e}{\mathbb{E}}

\newcommand{\bmf}{\boldsymbol{\mathcal{F}}}

\newcommand{\bfi}{\bm{\phi}}
\newcommand{\dkl}{D_{\textup{KL}}}

\renewcommand{\qed}{$\hfill\blacksquare$}
\DeclareMathOperator{\asc}{\xrightarrow{\textup{a.s.}}}
\DeclareMathOperator{\asceq}{\stackrel{\textup{a.s.}}{=}}
\DeclareMathOperator{\T}{\mathsf{T}}

\newtheorem{theorem}{Theorem}
\newtheorem{lemma}{Lemma}
\newtheorem{proposition}{Proposition}
\newtheorem{corollary}{Corollary}
\newtheorem{assumption}{Assumption}
\newtheorem{definition}{Definition}

\begin{document}

\title{Social Opinion Formation and Decision Making Under Communication Trends}

\author{\vspace{0.8em}\IEEEauthorblockN{Mert Kayaalp, Virginia Bordignon, Ali H. Sayed \vspace{0.5em}\\
\textit{\'{E}cole Polytechnique F\'{e}d\'{e}rale de Lausanne (EPFL)}}
    \thanks{The authors are with Adaptive Systems Laboratory, EPFL. Emails: \{mert.kayaalp, virginia.bordignon, ali.sayed\}@epfl.ch. This work was supported in part by SNSF grant 205121-184999.}}

\maketitle

\begin{abstract}
  This work studies the learning process over social networks under partial and random information sharing. In traditional social learning models, agents exchange full belief information with each other while trying to infer the true state of nature. We study the case where agents share information about only one hypothesis, namely, the trending topic, which can be randomly changing at every iteration. We show that agents can learn the true hypothesis even if they do not discuss it, at rates comparable to traditional social learning. We also show that using one's own belief as a prior for estimating the neighbors' non-transmitted beliefs might create opinion clusters that prevent learning with full confidence. This phenomenon occurs when a single hypothesis corresponding to the truth is exchanged exclusively during all times. Such a practice, however, avoids the complete rejection of the truth under any information exchange procedure --- something that could happen if priors were uniform.
  
\end{abstract}

\begin{IEEEkeywords}
social learning, distributed inference, distributed hypothesis testing, diffusion strategy, trending topics, partial information sharing
\end{IEEEkeywords}

\section{Introduction}
\IEEEPARstart{S}{ocial} learning \cite{chamley_2003,djuric2012,chamley2013,krishnamurthy_2013,mossel2017opinion} models opinion formation and decision making by agents connected by a graph topology. In these models, agents observe data and interact with their neighbors in order to infer the true state of nature from a finite set of hypotheses. In a behavioral context, for example, voters may seek to agree on the best political representative among a set of candidates \{A, B, C, D\}, based on their personal biases as well as on their interactions over a social network. Other examples arise in the context of engineering systems: A sensor network may be cooperating to detect whether the weather is sunny or rainy or to classify imagery captured from a common scene \cite{bordignon2021learning}.

\begin{figure}[ht]
\centerline{\includegraphics[width=0.9\linewidth]{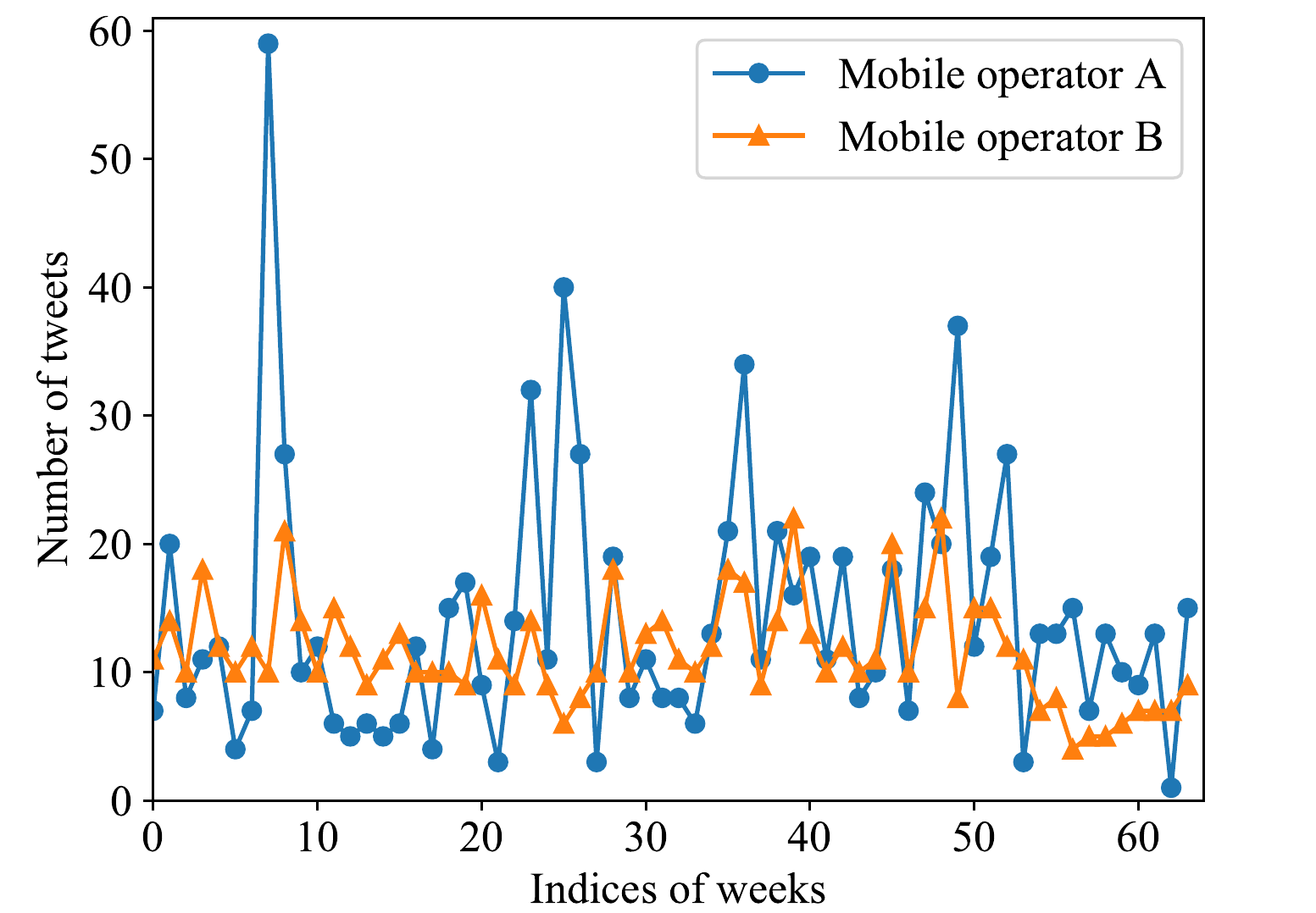}}
\caption{Number of tweets about two mobile operators for each week. (Mobile operator A- Swisscom, B- SunRise). Total of 64 weeks between 01/01/2021 and 30/04/2022, tweets in Switzerland.}
\label{fig:mobile_providers}
\end{figure}

In this work, our focus is on \emph{locally} Bayesian (a.k.a non-Bayesian) social learning strategies \cite{jadbabaie_2012,zhao_2012,nedic_2017,lalitha_2018,parasnis2020non}. While \emph{fully} Bayesian social learning \cite{acemoglu_2011} seeks to form the global Bayesian posterior, it necessitates extensive knowledge about other agents, such as their likelihood functions and network topology. Even under this extensive knowledge, achieving the global Bayesian posterior is known to be NP-hard \cite{hkazla2021bayesian}. In contrast, locally Bayesian social learning strategies rely solely on the localized processing of data and on localized interactions, and have been shown to allow the inference of the true state of nature \cite{jadbabaie_2012,zhao_2012}. Furthermore, these strategies are better suited to real-world scenarios compared to fully Bayesian strategies \cite{acemoglu_2011,hkazla2021bayesian} for at least two reasons. First, from a behavioral perspective, non-Bayesian strategies agree with the theory of bounded rationality in human decision making \cite{simon1990bounded,conlisk1996bounded}. Second, from an engineering perspective, they allow fully decentralized designs with moderate complexity and efficient memory \cite{jadbabaie_2012}. 

Non-Bayesian social learning algorithms repeatedly execute the following two steps: $i)$ agents update their \emph{local} beliefs, based on personal observations; $ii)$ agents combine their neighbors’ beliefs using a weighted averaging scheme like consensus \cite{degroot1974} or diffusion \cite{sayed_2014}. An implicit assumption common to these models is that agents are willing to share with neighbors their \emph{full} belief vector. That is, they share their beliefs about all possible hypotheses. In the context of social networks, this can be an unrealistic assumption. For instance, oftentimes, Twitter users concentrate on particular topics that constitute Twitter Trends. If candidate A gave a recent press release, Twitter users will likely focus on A when exchanging opinions and ignore other candidates. Another example can be users trying to determine the best mobile operator --- see Fig.~\ref{fig:mobile_providers} displaying the evolution of discussions on Twitter over time about two operators in Switzerland. The communication trends in this case can change based on the campaigns and advertisements by the mobile service providers (rather than being fixed over time as assumed by \cite{bordignon2020partial}). Furthermore, in the context of engineering systems, transmitting partial beliefs rather than full beliefs enable the design of communication-efficient systems under limited resources. Motivated by these examples, we are interested in the case where social agents share information on a random hypothesis of interest at each iteration.  \\

\noindent \textbf{Contributions.} 
\begin{itemize}
    \item We propose a social learning algorithm, where agents share their beliefs on \emph{only one} randomly chosen hypothesis at each time instant --- see Section \ref{sec:algorithm_description}. As opposed to the prior work \cite{bordignon2020partial}, the hypothesis being exchanged between agents is allowed to change over time. Moreover, under this partial information sharing scheme, agents complete the missing components of the received beliefs by using their own beliefs.
    \item When a wrong hypothesis is exchanged with positive probability, we show that beliefs evaluated at that hypothesis decay exponentially. The decay rate is the same as the asymptotic learning rate of traditional social learning algorithms (Theorem \ref{theorem:asymptotic_rate}). As a result, if each wrong hypothesis is exchanged with positive probability, then, learning occurs with probability one (Corollary \ref{corollary:truth_learning}). 
    \item We develop new proof techniques to tackle the randomness in the combination policy stemming from the shared hypothesis. The constraint that the beliefs belong to the probability simplex couples the processes for different hypotheses, making the standard application of the strong-law-of-large-numbers nonfeasible. Thus, we utilize martingale arguments to handle the non-linearity.
    \item We provide a counter-example in Section \ref{sec:truth_sharing} to show that sharing information about the true hypothesis is not sufficient for truth learning with full confidence when agents use their own beliefs as a prior for other agents' beliefs.
    \item On the other hand, Theorem \ref{th:impossible_mislearning} states that agents will never discard the truth completely. Namely, their beliefs on the true hypothesis will never be zero, which also means that they will never be fully confident on a wrong hypothesis being the true hypothesis. This contrasts with the findings in \cite{bordignon2020partial}, where truth sharing is shown to lead to truth learning when agents adopt uniform priors for the missing components. However, with such an approach, agents might also become fully confident in an incorrect hypothesis \cite{bordignon2020partial}.
    \item We support the theoretical findings with simulations in Section \ref{sec:simulations}.
\end{itemize}

\section{Related Work}

The connections between the current work and the distributed decision-making literature can be grouped as follows. \\

\noindent \textbf{Partial information sharing}: Partial information is considered in the context of social learning in~\cite{bordignon2020partial}, where each agent shares a single hypothesis of interest, which is {\em fixed} over time. This setup does not cover the situation in which the hypothesis of interest might change, which is relevant in real social network dynamics. In our recurrent political example, users can switch the discussion topic from candidate A to candidate B, if suddenly a new event brings relevance to B. To represent this {\em randomness}, we choose to model the hypothesis of interest, or \emph{trending} hypothesis, as a random variable following some underlying probability mass function over the set of hypotheses. Another notable work on social learning with partial information is \cite{salhab2020}, which assumes that each agent shares an \emph{action} (i.e., a sample) that is determined with respect to its belief, which is a different problem than the current one. \\

\noindent \textbf{Limited communication over networks}: There are also works whose purpose is to design communication-efficient distributed decision-making systems. For instance, the work \cite{mitraEvent} studies event-triggered information sharing mechanisms. In contrast, we consider a \emph{random} choice of the trending hypothesis independent of the history of the environment. In \cite{paritoshAssignment}, distributing partial hypotheses sets for agents to track is examined, but we do not consider hypothesis assignment problems in this work. In addition to these, the work \cite{inan2022social} proposes communicating with one randomly sampled agent instead of all neighbors at each time instant. Moreover, quantizing the beliefs is possible and studied in \cite{toghani2021communication, mitra2021}. In our work, decreasing the communication burden on the nodes is instead achieved by transmitting partial beliefs. \\

\noindent \textbf{Estimation of hidden belief components}: When agents share partial beliefs, one is faced with the challenge of estimating the hidden/latent belief components. In \cite{bordignon2020partial}, the hidden beliefs are assumed to be uniform among non-transmitted components. Under some situations, this strategy can become fully confident on a wrong hypothesis. In this work, we propose an algorithm where agents \emph{bootstrap}. Namely, they utilize their own beliefs for estimating their neighbors' non-transmitted components. As it is shown in the sequel, this practice helps make the system more robust by precluding learning a wrong hypothesis. Note that the concurrent work \cite{cirillo2022} extends \cite{bordignon2020partial} for memory-aware strategies for the fixed transmitted hypothesis case. In contrast, we study dynamic hypotheses and the strategy we propose in this work estimates the hidden belief components in a different way. \\

\noindent \textbf{Network disagreement}: In traditional social learning where agents share their entire belief with each other, it is known that all agents eventually come to an agreement on the truth \cite{jadbabaie_2012,zhao_2012,nedic_2017,lalitha_2018}. However, a full agreement across agents is rarely encountered over  social networks in real life. As we show in Sec.~\ref{sec:truth_sharing}, the property that agents utilize their own information while filling hidden components of the received beliefs, make them \emph{conservative}. This can prevent agreement across the network as well as truth learning. Note that this phenomenon is different from the \emph{stubborn} behavior considered in \cite{acemoglu13, yildiz13,lena2019,vial2021local}, where the stubborn agents are non-typical agents that do not change their beliefs and possibly supply misinformation to the network. Other works that provide proofs for network disagreement include $i)$ the work \cite{bordignon_2021,kayaalp_dist_bayesian}, which considers social learning under dynamic state of nature, and also $ii)$ the works \cite{blondel09,acemoglu2021}, which studies homophilic agents.   \\

\noindent \textbf{Decentralized optimization}: A related line of work is decentralized optimization over networks which has a rich literature (see e.g., \cite{nedic2009distributed,sayed_proc2014,sayed_2014,dimakis2010_gossip,sayed_2023}). Several scenarios for networked agents behavior are considered in the literature, including asynchronous behavior \cite{rabbat2014_asynchronous,zhao_2015_p1,zhao_2015_p2,zhao_2015_p3,lian_2018, bedi2019_asynchronous}, quantized communication \cite{cao_2020,hanna2021_quantization}, and randomized partial adaptation based on coordinate descent algorithms \cite{ghassemi_2015,wang2017_coordinate}. However, our study is on social learning, which amounts to a different problem formulation. For example, within a deep learning context, decentralized optimization corresponds to collaborative training of neural networks using labeled data, whereas social learning corresponds to a collective inference process, where agents aim to classify sequential and unlabeled data under partial and heterogeneous information \cite{hu2023_performance}. Another important concept from the optimization domain is the stochastic mirror descent algorithm \cite{yuan2018_mirror,sayed_2023} whose connection to canonical social learning algorithms is studied in \cite{nedic2016}. In Appendix~\ref{ap:mirror_descent}, we provide the connection to mirror descent for the current partial information sharing setting. \\ 
 
\noindent \textbf{Notation.} We denote random variables in bold, e.g., \( \bm{x}\). The expectation with respect to the distribution of $\bm{x}$ is denoted by $\e_x$. We use $i$ to denote time steps as opposed to some prior works that use it for agent indices. Almost sure convergence of a sequence of random variables \( \{ \bm{x}_i\}\) to a random variable \( \bm{x}\) is denoted by
\begin{equation}
    \lim_{i \to \infty} \bm{x}_i \asceq \bm{x}.
\end{equation}
When clear from the context, we simply write \( \bm{x}_i \asc \bm{x} \). $\dkl (p || q)$ represents the Kullback-Leibler (KL) divergence of two distributions $p$ and $q$. All-ones vector of dimension $K$ is denoted by $\mathds{1}_K$. The operator $\mathbb{I}$ denotes the indicator function such that when the Boolean condition $B$ is true, then $\mathbb{I}\{B\} = 1$. Otherwise, $\mathbb{I}\{B\} = 0$.

\section{Preliminaries}

\subsection{Problem Setting}
We consider the setting in which \( K \) agents over a network $\mathcal{N}$ try to distinguish the true state of nature \( \theta^\circ \in \Theta \), from among a finite set of \( H \) hypotheses, collected into the set \(\Theta = \{1,2,\dots,H \} \). The confidence of each agent \( k \) on \( \theta \in \Theta \) being the true hypothesis at time instant \( i \) is denoted by the belief \( \mu_{k,i} ( \theta ) \). The belief vector \( \mu_{k,i} \) is assumed to be a probability mass function (pmf) over all hypotheses $\Theta$. At each time instant \( i \), each agent \( k \) receives a private observation \( \bxi_{k,i} \) that is partially informative about the true state \( \theta^\circ \). This observation is distributed according to some marginal distribution \( L_k (\cdot | \theta^\circ ) \), with all observations assumed to be independent and identically distributed (i.i.d.) over time. However, observations can be dependent across different agents. As in other works on social learning, we assume that agents know their own marginal likelihood functions, denoted by \( L_k (\cdot | \theta)\) for all \( \theta \in \Theta \). The models by agent $k$ are only known to that agent. To ensure that the models share the same support with the unknown true distribution, it is assumed that:
\begin{align}\label{eq:assumption_finite_kl}
    \dkl ( L_k (\cdot | \theta^\circ ) ||  L_k (\cdot | \theta )) < \infty
\end{align}
for all agents and hypotheses, i.e., for all $k\in\mathcal{N}$ and $\theta\in\Theta$. This assumption ensures that each observation has finite informativeness about the true state. \\

\noindent \textbf{Network topology:} Agents are assumed to be connected by a strongly-connected graph \cite{sayed_2014} to which we associate a combination matrix $ A \triangleq [a_{\ell k}]$ --- see Fig.~\ref{fig:network_main}. In other words, there exists a path between any pair of agents \( (\ell,k )\), and there is at least one agent with a positive self-loop, say, $a_{mm}>0$ for some agent $m$. This implies that \( A \) is a primitive matrix \cite{sayed_2014}. Each entry \( a_{\ell k}\) of the matrix $A$ is nonnegative and represents the weight agent \( k \) assigns to the information it receives from agent \( \ell \). The scalar \( a_{\ell k}\) is nonzero if, and only if, agent \( \ell \) is a neighbor of agent \( k \), i.e., \( \ell \in \mathcal{N}_k \). The matrix \( A \) is assumed to be left stochastic; i.e., the entries on each of its columns add up to 1. In other words,
\begin{equation}
\mathds{1}_{K}^{\T}A=\mathds{1}_{K}^{\T},\quad a_{\ell k} \geq 0, \quad a_{\ell k }=0 \iff \ell\notin\mathcal{N}_k. 
\end{equation}
It follows from the Perron-Frobenius theorem \cite{sayed_2014,sayed_2023} that $A$ has an eigenvector corresponding to the eigenvalue at 1, and all entries of this eigenvector are positive and add up to one, i.e.,
\begin{equation}\label{eq:perron_def}
Av = v,\quad \mathds{1}_{K}^{\T} v = 1,\quad v\succ 0,
\end{equation}
where $\succ$ denotes an element-wise inequality. We refer to $v$ as the Perron vector of $A$. Its $k$-th entry \(v_k\) is a measure of the centrality of agent $k$.

\begin{figure}[ht]
\centerline{\includegraphics[width=0.65\linewidth]{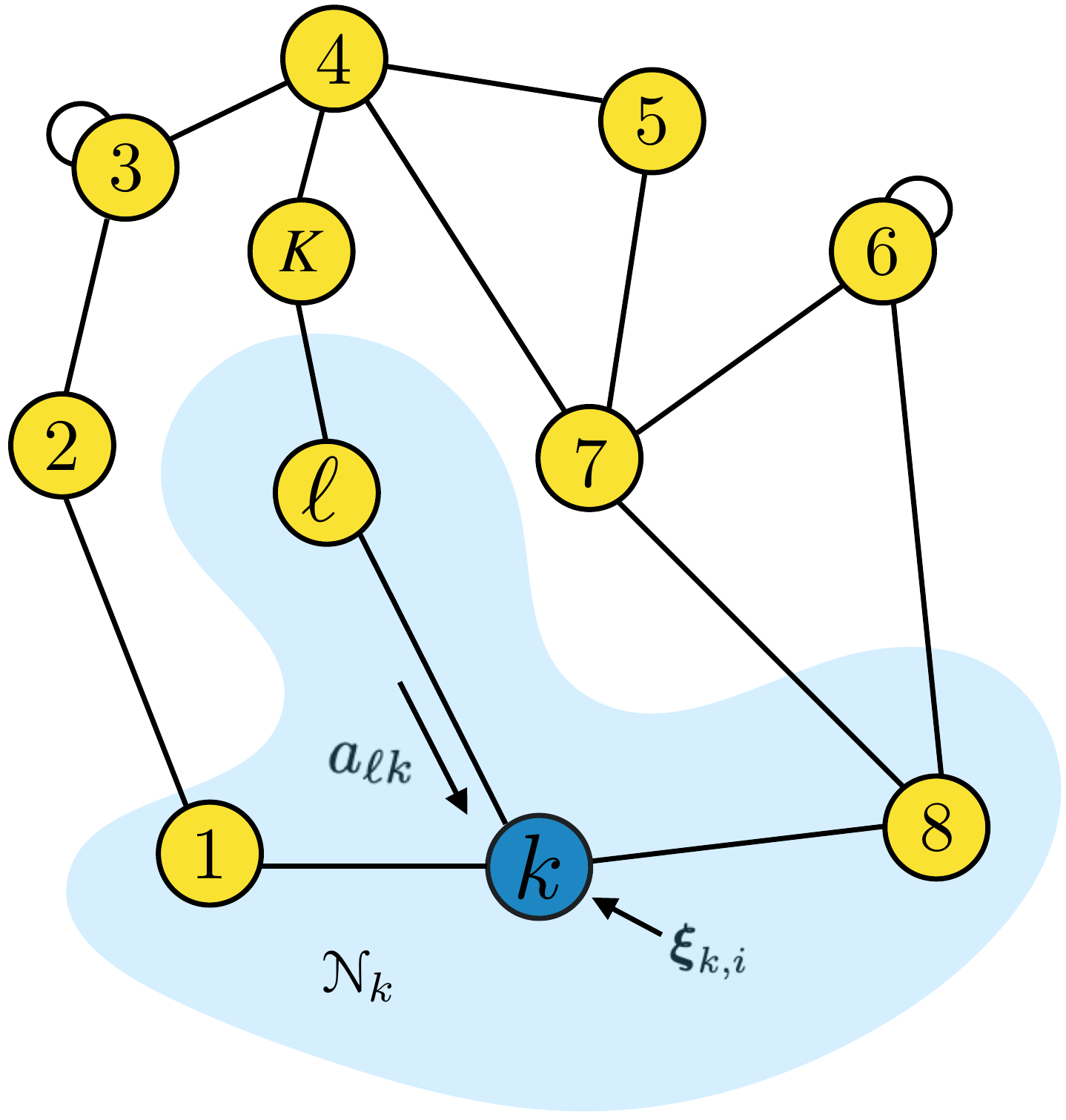}}
\caption{An illustration of the network model. While the network is in general directed, for visual simplicity, the edges are represented with lines instead of arrows in this figure.}
\label{fig:network_main}
\end{figure}

In this work, we are interested in whether individual agents can distinguish the true state $\theta^\circ$ from other hypotheses. We distinguish three possible scenarios on this matter.
\begin{definition}[\textbf{Learning and mislearning}]\label{def:learning_mis}
We say that truth learning occurs whenever:
\begin{equation}
    \bmu_{k,i} (\theta^\circ) \asc 1, \qquad \forall k \in \mathcal{N}
\end{equation}
In other words, truth learning means that agents become fully confident on the true hypothesis with probability one. In any other scenario we say that there is no learning. Among the cases where there is no learning, we define total mislearning as corresponding to the case in which for some \( \theta \in \Theta \setminus \{ \theta^\circ \}\):
\begin{equation}
     \bmu_{k,i} (\theta) \asc 1, \qquad \forall k \in \mathcal{N}
\end{equation}
In this case, agents become fully confident on a wrong hypothesis, with probability one.\qed
\end{definition}

We emphasize that according to Definition~\ref{def:learning_mis}, learning and mislearning occur whenever agents become fully confident on some hypothesis. In traditional Bayesian learning, the belief incorporates knowledge from an increasing number of measurements, with the expectation that the belief is not only maximized around the truth but also concentrates its mass on it. Any other outcome would imply a system malfunction. The same asymptotic learning outcome is expected in our scenario as well. An infinite amount of evidence should lead to certainty. An alternative definition of learning could be to require that the agents' beliefs are maximized at the true hypothesis \cite{bordignon_2021}. We adopt the former definition for consistency with other traditional works in the field.

Before presenting the main contributions of this work, for the benefit of reader, we briefly review the traditional social learning paradigms with full and partial information sharing \cite{jadbabaie_2012,zhao_2012,nedic_2017,lalitha_2018,bordignon2020partial}.

\subsection{Full Information Sharing}

In non-Bayesian social learning, at every iteration, agents first incorporate their personal observations to their beliefs. To do so, they perform a \emph{local} Bayesian update and obtain intermediate beliefs:
\begin{align}\label{eq:local_bayesian}
    \bpsi_{k,i} (\theta) = \frac{L_k (\bxi_{k,i} | \theta ) \bmu_{k,i-1}(\theta)}{\sum_{\theta^\prime \in \Theta}L_k (\bxi_{k,i} | \theta^\prime ) \bmu_{k,i-1}(\theta^\prime)}.
\end{align}
Subsequently, they exchange the intermediate beliefs with their immediate neighbors in the graph for \emph{all} hypotheses. Each agent then averages geometrically the received intermediate beliefs from its neighbors to obtain the updated belief:
\begin{equation}\label{eq:trad_soc_combine}
    \bmu_{k,i}(\theta) = \frac{\prod\limits_{\ell \in \mathcal{N}_k} (\bpsi_{\ell,i} (\theta))^{a_{\ell k}} }{\sum\limits_{\theta^\prime \in \Theta} \prod\limits_{\ell \in \mathcal{N}_k} (\bpsi_{\ell,i} (\theta^\prime))^{a_{\ell k}} },
\end{equation}
for each \( \theta \in \Theta\). In \eqref{eq:trad_soc_combine}, the term in the numerator is a weighted geometric average of the received intermediate beliefs, and the denominator is a normalization term that makes sure the belief vector \( \bmu_{k,i} \) is a pmf over \(\Theta\). Note that there are variations where the combination step \eqref{eq:trad_soc_combine} uses an arithmetic averaging operation, as opposed to geometric averaging \cite{jadbabaie_2012,zhao_2012}. These two averaging strategies possess distinct properties \cite{koliander2022} and the choice of using which for social learning is dependent upon the specific application under consideration. However, in general, it is more common to use geometric averaging for fusing probability density functions, while arithmetic averaging is typically favored for combining point estimates of random variables \cite{li2019second}. Moreover, in the context of social learning, the work \cite{kayaalp2022_aa_ga} demonstrated that geometric averaging yields a faster rate learning compared to the arithmetic averaging. Therefore, in this work, we focus on geometric averaging strategies as in \cite{nedic_2017,lalitha_2018}.

To avoid trivial cases where agents discard some hypotheses right from the start, the following condition is necessary.
\begin{assumption}[\bf{Initial beliefs}]\label{as:positive_initial_beliefs}
All initial beliefs are strictly positive at all hypotheses, i.e., for each agent \( k \) and for all hypotheses \( \theta \in \Theta \), \( \bmu_{k,0} (\theta) > 0\).
\qed
\end{assumption}
\noindent If Assumption~\ref{as:positive_initial_beliefs} is violated for a hypothesis \(\theta\), due to the nature of the geometric fusion rule, all agents would have zero belief on \(\theta\) in finite time. Moreover, the following condition enables the aggregate of all agents to distinguish the true hypothesis from the wrong ones.
\begin{assumption}[\textbf{Global identifiability}]\label{as:global_iden}
For each wrong hypothesis \( \theta \neq \theta^\circ \), there exists at least one agent \( k \) that can distinguish \( \theta \) from \( \theta^\circ \), namely,
\begin{align}
    \dkl\Big (  L_k (\cdot | \theta^\circ) \big |\big| L_k (\cdot | \theta) \Big ) > 0
\end{align}\qed
\end{assumption}
\noindent Note that this is a weaker assumption than \emph{local} identifiability, which requires \emph{all} agents (not only one) to have the capability of distinguishing the truth individually. Under Assumptions \ref{as:positive_initial_beliefs} and \ref{as:global_iden}, agents can learn the truth with full confidence.
\begin{proposition}[\textbf{Learning with full communication \cite{nedic_2017,lalitha_2018}}]\label{prop:full_com_learning}
Under Assumptions \ref{as:positive_initial_beliefs} and \ref{as:global_iden}, the social learning algorithm \eqref{eq:local_bayesian}--\eqref{eq:trad_soc_combine} allows each agent $k$ to learn the truth with probability one, i.e.,
\begin{equation}\label{eq:ga_standard_truth_learning}
    \bmu_{k,i} (\theta^\circ) \asc 1.
\end{equation}
Furthermore, the asymptotic learning of the truth occurs at an exponential rate. That is to say, the asymptotic convergence rate of the beliefs on wrong hypotheses $\theta \neq \theta^\circ$ to zero is given by
\begin{equation}\label{eq:ga_standard_rate}
       \frac{1}{i} \log \frac{ \bmu_{k,i}(\theta)}{ \bmu_{k,i}(\theta^\circ)} \asc \sum_{k = 1}^K  -v_k \dkl\Big (  L_k (\cdot | \theta^\circ) \big |\big| L_k (\cdot | \theta) \Big ).
\end{equation}
\end{proposition}
\begin{proof}
The full proof appears in \cite{nedic_2017,lalitha_2018}. In Appendix~\ref{ap:ful_com_learning}, we provide a brief sketch for the benefit of the reader. 
\end{proof}

 The future results stated in Theorem~\ref{theorem:asymptotic_rate} and Corollary~\ref{corollary:truth_learning} will show that agents can still learn the truth if they exchange beliefs on \emph{only one} random hypothesis at each time instant. The randomness in the shared hypothesis and the incomplete information received from the neighbors introduce new challenges into the analysis, in comparison to the proof of Proposition~\ref{prop:full_com_learning} above. This aspect is treated in later sections.
\subsection{Partial Information Sharing}
The canonical setting described in the previous section assumes that agents share their \emph{entire}
beliefs with each other at every iteration. The work \cite{bordignon2020partial} considers an alternative setting where agents share their beliefs on only one hypothesis of interest $\tau \in \Theta $, which is \emph{fixed} over time. In particular, after agents perform the local adaptation step \eqref{eq:local_bayesian}, each agent $k$ receives $\bpsi_{\ell,i} (\tau)$ from its neighbors $\ell\in\mathcal{N}_k$. Then, each agent $k$ completes the missing entries of the received belief vectors by using the construction:
\begin{align}\label{eq:partial_information}
         \widehat{\bpsi}_{\ell,i} (\theta) = \begin{dcases} 
      \bpsi_{\ell,i}(\tau), & \theta = \tau \\
      \frac{1-\bpsi_{\ell,i}(\tau)}{H-1}, & \theta \neq \tau
   \end{dcases}
\end{align}
for $\ell\in\mathcal{N}_k$. Observe that the received component of the intermediate belief is used as is, while the remaining components are assigned uniform weight. Then, these modified beliefs are used in the combination step:
\begin{equation}\label{eq:part_soc_combine}
    \bmu_{k,i}(\theta) = \frac{\prod\limits_{\ell \in \mathcal{N}_k} (\widehat{\bpsi}_{\ell,i} (\theta))^{a_{\ell k}} }{\sum\limits_{\theta^\prime \in \Theta} \prod\limits_{\ell \in \mathcal{N}_k} (\widehat{\bpsi}_{\ell,i} (\theta^\prime))^{a_{\ell k}} }.
\end{equation}
The work \cite{bordignon2020partial} provides a detailed analysis of this algorithm and proposes a self-aware variant of it where each agent $k$ uses $\bpsi_{k,i} (\theta)$ instead of the approximation $\widehat{\bpsi}_{k,i} (\theta)$ in \eqref{eq:part_soc_combine}. For the purposes of the current work, it is enough to state the following results, which we will refer to in the sequel. Let us introduce the following notation for the average likelihood function. The average is over the non-transmitted components (denoted by the symbol $\overline{\tau}$):
\begin{equation}
    L_k (\xi | \overline{\tau}) \triangleq \frac{1}{H-1} \sum_{\theta \in \Theta \setminus \{\tau\}} L_k (\xi | \theta).
\end{equation}
Then, the following two results are from \cite{bordignon2020partial}.
\begin{proposition}[\textbf{Learning with truth sharing \cite[Theorem 1]{bordignon2020partial}}]\label{prop:part_information_learning}
 If \( \tau = \theta^\circ\), under Assumption~\ref{as:positive_initial_beliefs} and a modified global identifiability assumption, i.e., if there exists an agent such that
    \begin{equation}
    \dkl\Big (  L_k (\cdot | \tau) \big |\big| L_k (\cdot | \overline{\tau}) \Big ) > 0,
    \end{equation}
then, truth learning occurs with full confidence:
\begin{equation}
     \bmu_{k,i} (\tau) \asc 1
\end{equation}
    for each agent $k$ under \eqref{eq:partial_information}--\eqref{eq:part_soc_combine}. \qed
\end{proposition}
Proposition~\ref{prop:part_information_learning} shows that if the hypothesis $\tau$ agents exchange information about happens to be the truth $\theta^\circ$, then this is sufficient for truth learning with full confidence. However, if they are not discussing the truth, this can lead to mislearning a wrong hypothesis as the next result illustrates.
\begin{proposition}[\textbf{Total mislearning \cite[Theorem 3]{bordignon2020partial}}]\label{prop:part_information_mislearning}
 If \( \tau \neq  \theta^\circ\), under Assumption~\ref{as:positive_initial_beliefs} and the condition that
    \begin{equation}
     \sum_{k=1}^K \!v_k \dkl\Big (\!  L_k (\cdot | \theta^\circ) \big |\big| L_k (\cdot | \tau) \!\Big ) \!<\! \sum_{k=1}^K \!v_k \dkl\Big ( \! L_k (\cdot | \theta^\circ) \big |\big| L_k (\cdot | \overline{\tau})\! \Big )
    \end{equation}
all agents learn a wrong hypothesis with full confidence:
\begin{equation}
     \bmu_{k,i} (\tau) \asc 1
\end{equation}
    for each agent $k$ under \eqref{eq:partial_information}--\eqref{eq:part_soc_combine}. \qed
\end{proposition}
In the algorithm we propose in this work (Alg.~\ref{alg:rand_inf_sharing}), agents will fill the missing beliefs using their own information, instead of assigning uniform values to them. This practice will result in outcomes that are opposite of Propositions \ref{prop:part_information_learning} and \ref{prop:part_information_mislearning}. Namely, we show in the following that exclusive truth sharing does not suffice for truth learning (Sec.~\ref{sec:truth_sharing}), while total mislearning can never occur (Theorem~\ref{th:impossible_mislearning}) under Algorithm~\ref{alg:rand_inf_sharing}. 

\section{Social Learning under Trending Topics} \label{sec:algorithm_description}
As opposed to existing works that require transmission of the entire beliefs at each iteration, or exchanging a fixed component of the beliefs, in this work we allow agents to share a random component at each iteration. Similar to \cite{bordignon2020partial}, each agent will continue to compute its intermediate belief $\bm{\psi}_{k,i}(\theta)$ for every possible $\theta\in\Theta$ according to \eqref{eq:local_bayesian}, and the agents will continue to share with their neighbors information about their belief for \emph{only one} of the hypotheses. However, the agents will be sharing information about the same \emph{trending} hypothesis \( \btau_i \in \Theta \). Here, $\btau_i$ is a random variable that depends on the time instant $i$, and is therefore allowed to change over time. The motivation for this setting is two-fold. First, people tend to concentrate on particular topics of discussion over social networks, as indicated earlier in the introduction. Second, if agents (e.g., sensors, robots) have the same random key, they can choose the same hypothesis at each iteration without needing a central controller, and then exchange beliefs on that hypothesis alone, as opposed to the more costly approach of exchanging entire beliefs. Utilizing randomness this way can prove useful compared to a periodic scheduling scheme for security reasons, since deciphering a random key is harder for eavesdroppers than inferring a fixed schedule. We denote the distribution of $\btau_i$ by \( \pi \), where $\pi$ is a probability vector over the set of hypotheses $\Theta$, and write \( \mathbb{P}( \btau_i = \theta)=\pi_\theta \). We further assume that $\btau_i$ is i.i.d. over time and also independent of all observations over space and time. Again, since the agents receive incomplete belief vectors from their neighbors (actually, they receive only one entry from these vectors), the agents will again need to complete the missing entries. In the proposed strategy, the agents use their own intermediate local beliefs to fill in for the missing beliefs from their neighbors by using the following construction. Namely, agent $k$ completes the belief vector received from its neighbor $\ell$ by using:
\begin{align}\label{eq:approximate_int_phi}
     \bfi^{(k)}_{\ell,i} (\theta) = \begin{dcases} 
      \bpsi_{\ell,i}(\theta), & \theta = \btau_i \\
     \bpsi_{k,i}(\theta), & \theta \neq \btau_i
   \end{dcases}.
\end{align}
By doing so, the entries of $\bfi^{(k)}_{\ell,i}$ need not add up to 1. For this reason, agent $k$ can normalize \eqref{eq:approximate_int_phi} according to:
\begin{equation}\label{eq:approximate_int}
    \widehat{\bpsi}^{(k)}_{\ell,i} (\theta) = \frac{\bfi^{(k)}_{\ell,i} (\theta)}{\sum_{\theta^\prime \in \Theta}\bfi^{(k)}_{\ell,i} (\theta^\prime)}.
\end{equation}
whose denominator can be written as 
\begin{equation}
    \sum_{\theta^\prime \in \Theta}\bfi^{(k)}_{\ell,i} (\theta^\prime) = 1-\bpsi_{k,i}(\btau_i)+\bpsi_{\ell,i}(\btau_i).
\end{equation}
We refer to (19)--(20) as a bootstrapping step. Subsequently, the agents combine the \emph{approximate} intermediate beliefs from \eqref{eq:approximate_int} to update their beliefs as in \eqref{eq:trad_soc_combine} and \eqref{eq:part_soc_combine}:
\begin{align}\label{eq:combine}
\bmu_{k,i}(\theta) = \frac{\prod\limits_{\ell \in \mathcal{N}_k} (\widehat{\bpsi}_{\ell,i}^{(k)} (\theta))^{a_{\ell k}} }{\sum\limits_{\theta^\prime \in \Theta} \prod\limits_{\ell \in \mathcal{N}_k} (\widehat{\bpsi}_{\ell,i}^{(k)} (\theta^\prime))^{a_{\ell k}} }.
\end{align}
    The procedure is summarized in Algorithm~\ref{alg:rand_inf_sharing}. Its connection to distributed stochastic mirror descent algorithms is explained in Appendix~\ref{ap:mirror_descent}. Observe that the algorithm is \emph{self-aware} because for each agent $k$, it follows from \eqref{eq:alg_approximate} that \( \widehat{\bpsi}^{(k)}_{k,i} (\theta) = \bpsi_{k,i} (\theta)\). In other words, agents use their own intermediate beliefs as is.

  \begin{algorithm}[]
 \caption{Social learning with trending hypothesis}
 \begin{algorithmic}[1]
    \item set initial priors $\bmu_{k,0}(\theta)>0$,  $\forall \theta \in\Theta$ and $\forall k \in \mathcal{N}$
    \While{$i\geq 1$} 
 \For{each agent $k$} 
 \State receive private observation $\bxi_{k,i}$ \State \textbf{adapt} to obtain intermediate belief:
\begin{equation}\label{eq:alg_adapt}
\bpsi_{k,i} (\theta) = \frac{L_k (\bxi_{k,i} | \theta ) \bmu_{k,i-1}(\theta)}{\sum_{\theta^\prime \in \Theta}L_k (\bxi_{k,i} | \theta^\prime ) \bmu_{k,i-1}(\theta^\prime)}
\end{equation}
\EndFor
\State agents exchange \( \{\bpsi_{k,i} (\btau_i)\} \) on the current hypothesis
\Statex \quad \:\:of interest $\btau_i \sim \pi$
 \For{each agent $k$} 
 \State \textbf{approximate} intermediate beliefs for $\ell \in \mathcal{N}_k$ by \textbf{bootstrapping}:
\begin{align}\label{eq:alg_approximate}
        \quad \widehat{\bpsi}^{(k)}_{\ell,i} (\theta) = \begin{dcases} 
      \frac{\bpsi_{\ell,i}(\theta)}{1-\bpsi_{k,i}(\btau_i)+\bpsi_{\ell,i}(\btau_i)}, & \theta = \btau_i \\
      \frac{\bpsi_{k,i}(\theta)}{1-\bpsi_{k,i}(\btau_i)+\bpsi_{\ell,i}(\btau_i)}, & \theta \neq \btau_i
   \end{dcases}
\end{align}
\State \textbf{combine} approximate beliefs: 
\begin{equation}\label{eq:alg_combine}
    \bmu_{k,i}(\theta) = \frac{\prod\limits_{\ell \in \mathcal{N}_k} (\widehat{\bpsi}_{\ell,i}^{(k)} (\theta))^{a_{\ell k}} }{\sum\limits_{\theta^\prime \in \Theta} \prod\limits_{\ell \in \mathcal{N}_k} (\widehat{\bpsi}_{\ell,i}^{(k)} (\theta^\prime))^{a_{\ell k}} }
\end{equation}
\EndFor
\State \( i \leftarrow i+1\)
\EndWhile
 \end{algorithmic} \label{alg:rand_inf_sharing}
 \end{algorithm}

Recall that in traditional non-Bayesian social learning \eqref{eq:local_bayesian}--\eqref{eq:trad_soc_combine}, entire belief vectors are exchanged and hence there is no approximation of the intermediate beliefs, that is to say,
\begin{align}
    \widehat{\bpsi}^{(k)}_{\ell,i} (\theta) = \bpsi_{\ell,i}(\theta).
\end{align}
In addition, in the fixed hypothesis sharing case \eqref{eq:partial_information}, there is a fixed transmitted hypothesis \( \btau_i = \tau \) and non-transmitted hypotheses are assumed to be uniformly likely. In contrast, in \eqref{eq:alg_approximate} agents exploit their own beliefs as prior information.

\section{Main Results}

\subsection{Truth Learning}\label{sec:truth_learning}

In this section, we present results that characterize truth learning under certain conditions. First, we define the following loss function in order to assess the disagreement between the truth and the belief of each agent \( k \) at time \( i \):
\begin{align}
    Q(\bmu_{k,i}) \triangleq \dkl (\delta_{\theta^\circ} || \bmu_{k,i}),
\end{align}
where we denote the true probability mass function as the Kronecker-delta function:
\begin{align}
         \delta_{\theta^\circ} (\theta) = \begin{dcases} 
      1, & \theta = \theta^\circ \\
      0, & \theta \neq \theta^\circ
   \end{dcases}.
\end{align}
Observe that
\begin{equation}
    Q(\bmu_{k,i}) = \sum_{\theta \in \Theta} \delta_{\theta^\circ} (\theta) \log \frac{\delta_{\theta^\circ} (\theta)}{\bmu_{k,i} (\theta)} = -\log \bmu_{k,i} (\theta^\circ),
\end{equation}
where we use the convention that \( 0\log0=0\). The network loss is defined as the weighted average of the individual loss functions, where the weighting is given by the Perron entries:
\begin{equation}
    Q(\bmu_i) \triangleq \sum_{k=1}^K v_k Q(\bmu_{k,i})  =-\sum_{k=1}^K v_k \log \bmu_{k,i} (\theta^\circ) .
\end{equation}
We denote the history of observations and transmitted hypotheses up to time $i$ by $\bmf_i~\triangleq~\{ \btau_i, \bxi_i,\btau_{i-1},\bxi_{i-1},\dots\}$, where we introduced the aggregate vector of observations $\bxi_i \triangleq \{ \bxi_{k,i}\}_{k=1}^K$. The following result shows that the conditional expectation of the network loss does not increase, given the history $\bmf_{i-1}$.

\begin{lemma}[\bf{Network average loss}]\label{lemma:submartingale}
The network loss \( Q(\bmu_i) \) is a super-martingale, namely,
\begin{align}\label{eq:martingale_lemma_eq}
    \e \Big [ Q(\bmu_{i}) \Big | \bmf_{i-1} \Big ] &\leq Q(\bmu_{i-1})
\end{align}
\end{lemma}
\begin{proof}
See Appendix~\ref{ap:submartingale}.
\end{proof}

In view of this lemma, Algorithm~\ref{alg:rand_inf_sharing} leads to a robust design in the sense that given any particular network loss realization, the loss in the next iteration will not increase in expectation. Note that this result holds for any possible transmission distribution \(\pi\). Similar to Eq.~\eqref{eq:ga_standard_rate} for traditional social learning, we continue our presentation by focusing on the decay rate of the belief ratio for some wrong hypothesis \( \theta \in \Theta \setminus  \{ \theta^\circ\}\) under the proposed strategy \eqref{eq:alg_adapt}--\eqref{eq:alg_combine}.
\begin{theorem}[\textbf{Asymptotic learning rate}]\label{theorem:asymptotic_rate}
For each wrong hypothesis \( \theta \in \Theta \setminus \{ \theta^\circ \}\), if the transmission probability is strictly positive, i.e.,  \( \pi_\theta > 0 \), then the belief on that wrong hypothesis will converge to zero at an asymptotically exponential rate under Assumption~\ref{as:positive_initial_beliefs}. Namely, for each agent \( k \):
\begin{align}\label{eq:asym_learning_rate_ris}
        \frac{1}{i} \log \frac{ \bmu_{k,i}(\theta)}{ \bmu_{k,i}(\theta^\circ)} \asc \sum_{k = 1}^K  -v_k \dkl\Big (  L_k (\cdot | \theta^\circ) \big |\big| L_k (\cdot | \theta) \Big ).
\end{align}
\end{theorem}
\begin{proof}
See Appendix~\ref{ap:rate_convergence}.
\end{proof}
The convergence rate expression in \eqref{eq:asym_learning_rate_ris} is an average of the local KL-divergences of agents, which measures their individual informativeness, weighted by their network centralities (i.e., Perron entries). Observe that this rate matches the rate \eqref{eq:ga_standard_rate} for traditional social learning algorithms, which require transmission of full beliefs. Therefore, a positive probability of transmitting the wrong hypothesis suffices for achieving the same asymptotic performance with probability one, regardless of the transmission probabilities for the remaining hypotheses. Combining Theorem \ref{theorem:asymptotic_rate} and Assumption \ref{as:global_iden} directly yields the following sufficient conditions for truth learning.
\begin{corollary}[\textbf{Truth learning}]\label{corollary:truth_learning}
Under Assumptions \ref{as:positive_initial_beliefs} and \ref{as:global_iden}, if \( \pi_\theta > 0\) for all wrong hypotheses \(\theta \in \Theta \setminus \{ \theta^\circ \} \), then each agent \( k \) learns the truth with probability one, i.e.,
\begin{align}
    \bmu_{k,i}(\theta^\circ) \asc 1.
\end{align}\qed
\end{corollary}
Notice that any asymmetry between the entries of \( \pi \) does not affect the learning. In particular, more or less frequent communication of a hypothesis does not change the asymptotic rate of convergence. In \cite{bordignon2020partial}, truth learning (in the sense of Definition~\ref{def:learning_mis}) occurs if, and only if, the fixed transmitted hypothesis is the true hypothesis. Corollary \ref{corollary:truth_learning} shows that if agents are bootstrapping as opposed to using uniform weights \cite{bordignon2020partial}, then learning can occur as long as \( \pi_\theta > 0 \) for all wrong hypotheses \( \theta \). This implies that they can learn the truth even if they do not discuss the true hypothesis, i.e. even if \( \pi_{\theta^\circ} = 0 \).

\subsection{Truth Sharing}\label{sec:truth_sharing}

In the previous section, we drew the following two conclusions: In Theorem~\ref{theorem:asymptotic_rate}, we established that if there is a positive probability \( \pi_\theta > 0 \) of transmitting a wrong hypothesis \( \theta \neq \theta^\circ \), this is sufficient to reject $\theta$, i.e., beliefs on that hypothesis will go to zero exponentially fast. Building upon this, in Corollary~\ref{corollary:truth_learning} we showed that exchanging all wrong hypotheses with positive probability (i.e., \( \pi_\theta > 0, \: \forall \theta \neq \theta^\circ \)) is sufficient for agents to infer the truth. Given these findings, the question arises: Is the \emph{exclusive} exchange (\( \pi_{\theta^\circ}= 1 \)) of the true hypothesis alone enough for learning? We give a negative answer to this question by providing a toy counter-example where agents do not learn even when \( \pi_{\theta^\circ}= 1 \).

Consider a fully-connected network of 3 agents (see Fig. \ref{fig:three_agents}). The hypotheses set is \( \Theta = \{1,2,3,4 \}\) where incidentally \( \theta^\circ = 4 \). Assume that agent \( k \) cannot distinguish between the true hypothesis \( \theta^\circ = 4 \) and the hypothesis \( \theta_k = k \), i.e., \(  \dkl (  L_k (\cdot | \theta^\circ)  || L_k (\cdot | \theta_k)  ) = 0\). Assume further that each agent $k$ is capable of distinguishing hypotheses other than $\theta_k$ and $\theta^\circ$ (e.g., agent $1$ can distinguish hypotheses $2$ and $3$, but cannot distinguish $1$ and $4$). Since each wrong hypothesis \( \theta \neq \theta^\circ \) can be distinguished by at least one agent, the problem is globally identifiable, satisfying Assumption~\ref{as:global_iden}. 

\begin{figure}[ht]
\centerline{\includegraphics[width=\linewidth]{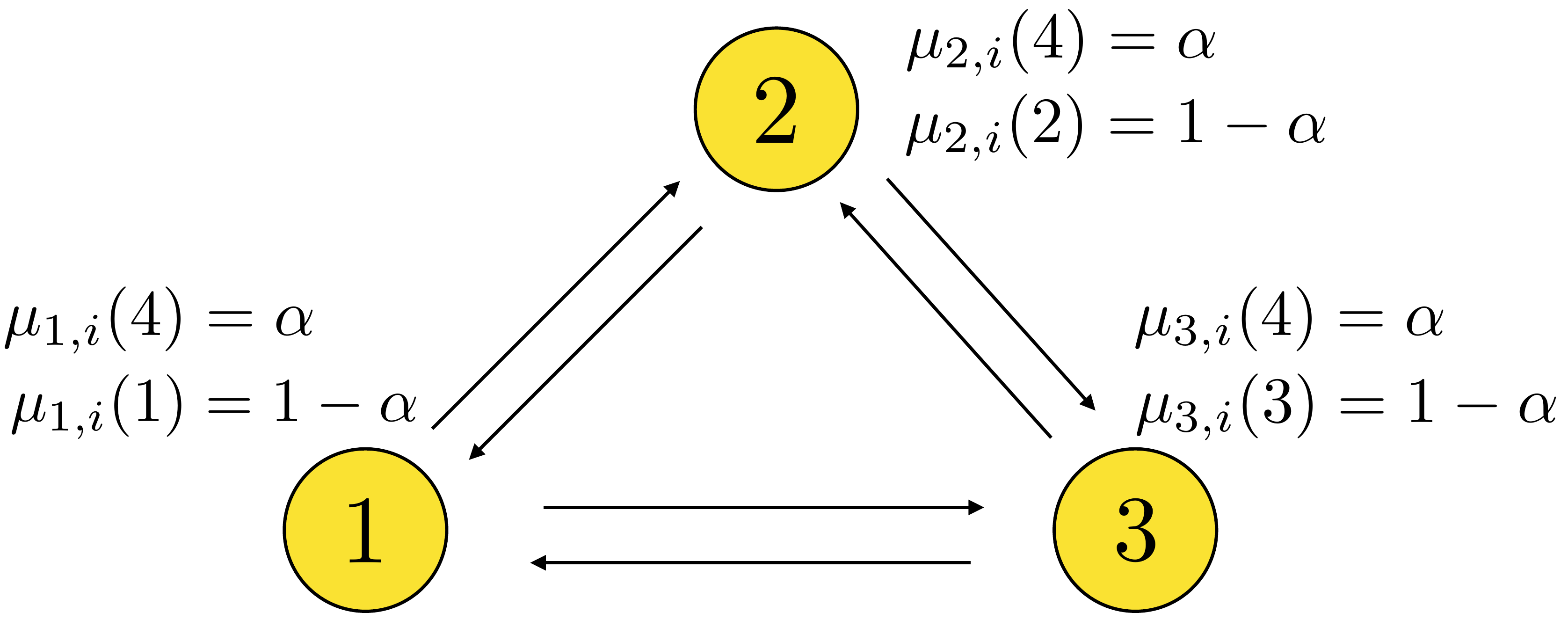}}
\caption{A network of 3 agents. Each agent $k$ can distinguish all but the hypotheses $\theta^\circ = 4$ and $\theta_k = k$. In the assumed scenario, belief values on $\theta^\circ = 4$ are assumed to be $\alpha \in (0,1)$ for all agents. Belief values at locally distinguishable hypotheses are assumed to be 0.}
\label{fig:three_agents}
\end{figure}

Imagine that at time \( i \), \( \bmu_{k,i} (\theta^\circ) = \alpha \) and \( \bmu_{k,i} (\theta_k) = 1-\alpha \) for each agent \( k \). In other words, for each agent $k$, the beliefs on hypotheses that are locally distinguishable are equal to 0 \footnote{Under Assumption~\ref{as:positive_initial_beliefs} and Eq. \eqref{eq:assumption_finite_kl}, the belief values remain nonzero for all finite time instants. For the ease of the presentation, we assume that sufficient time has elapsed so that the beliefs on locally distinguishable hypotheses are sufficiently close to 0.}. Using \eqref{eq:local_bayesian} yields the intermediate beliefs:
\begin{align}
    \bpsi_{k,i+1} (\theta^\circ) =\frac{L_k(\bxi_{k,i+1} | \theta^\circ) \bmu_{k,i} (\theta^\circ)}{L_k(\bxi_{k,i+1} | \theta^\circ) \bmu_{k,i} (\theta^\circ)+ L_k(\bxi_{k,i+1} | \theta_k) \bmu_{k,i} (\theta_k)}.
\end{align}
Since \(  L_k(\bxi_{k,i+1} | \theta^\circ) =  L_k(\bxi_{k,i+1} | \theta_k)\) (their KL divergence is assumed to be 0):
\begin{align}
    \bpsi_{k,i+1} (\theta^\circ) = \frac{\bmu_{k,i} (\theta^\circ)}{\bmu_{k,i} (\theta_k)+\bmu_{k,i} (\theta^\circ)} = \bmu_{k,i} (\theta^\circ) = \alpha
\end{align}
Similarly,
\begin{align}
    \bpsi_{k,i+1} (\theta_k) = \bmu_{k,i} (\theta_k) = 1- \alpha.
\end{align}
Agents fill the received intermediate beliefs by using their own beliefs (i.e, they bootstrap) according to \eqref{eq:alg_approximate}. Here, since \( \pi_{\theta^\circ} = 1\) (i.e., \(\btau_{i+1} = \theta^\circ\)) we get:
\begin{align}
    \widehat{\bpsi}^{(k)}_{\ell,i+1} (\theta^\circ) = \bpsi_{\ell, i+1} (\theta^\circ) = \alpha,
\end{align}
and
\begin{align}
    \widehat{\bpsi}^{(k)}_{\ell,i+1} (\theta_k) = \bpsi_{k, i+1} (\theta_k) =1- \alpha.
\end{align}
Finally, after the combination of the approximate intermediate beliefs \eqref{eq:alg_combine}, we arrive at
\begin{align}
    \bmu_{k,i+1} (\theta^\circ)= \bmu_{k,i} (\theta^\circ) = \alpha,
\end{align}
and
\begin{align}
    \bmu_{k,i+1} (\theta_k)= \bmu_{k,i} (\theta_k) = 1- \alpha.
\end{align}
This is an equilibrium (fixed point) for the algorithm. Regardless of the emitted observations, the beliefs of agents will not change over time. Consequently, if \( \alpha \) is small, agents can get stuck in beliefs where their confidence levels on the wrong hypotheses are higher than their confidence on the true hypothesis.

The results in Sec.~\ref{sec:truth_learning} and the counter-example in this section reveal that it is the exchange of the wrong hypotheses that promotes truth learning, and not truth sharing. The intuition behind this rather surprising outcome is the following. In the current strategy, agents fill the missing belief components with their own beliefs (using bootstrapping). If their beliefs happen to align closely on the true hypothesis, it is difficult to change them. This undesirable equilibrium is bypassed when agents exchange beliefs on wrong hypotheses. This is due to Assumption~\ref{as:global_iden}, which states that there exists at least one agent that is able to drive the beliefs on a wrong hypothesis to zero.

In \cite{bordignon2020partial} (see Eq. \eqref{eq:partial_information}), when the fixed transmitted hypothesis corresponds to the truth, truth learning in the sense of Definition~\ref{def:learning_mis} occurs almost surely. The example described in this section suggests that using one's own belief for estimating non-transmitted components of neighbors, i.e., bootstrapping, as opposed to using uniform priors, may lead regular agents to become \emph{conservative} about their own opinions. It can prevent learning under partial information sharing. In addition to not learning the truth, the network can also fail to reach consensus and opinion clusters might emerge. In Fig. \ref{fig:three_agents}, agents having positive beliefs for different hypotheses can have a major effect especially when \( \alpha \) is small. It leads to a strong network disagreement. Network disagreement phenomena were observed in the works \cite{acemoglu13, yildiz13,lena2019,vial2021local} when there are special agents who never change their opinions, i.e., stubborn agents. Our result, on the other hand, indicates that even when the network is only composed of regular agents, limited communication can hinder concurrence and truth learning.

\subsection{Impossibility of Mislearning}\label{sec:impossible_mislearning}

The previous section demonstrated that bootstrapping might induce network disagreement and poor equilibrium. In this section, we provide a positive result in the opposite direction: Agents will never be fully confident on a wrong hypothesis. Total mislearning cannot occur.

\begin{theorem}[\textbf{Impossibility of mislearning}]\label{th:impossible_mislearning}
Under Assumption \ref{as:positive_initial_beliefs}, agents will always have positive confidence on the true hypothesis. Namely, for each agent \( k \), $\bmu_{k,i}(\theta^\circ) > 0$ for any finite time $i$, and
\begin{align}
    \mathbb{P} \Big ( \liminf_{i \to \infty} \bmu_{k,i}(\theta^\circ) = 0 \Big ) = 0,
\end{align}
or alternatively, for \( \theta \in \Theta \setminus \{ \theta^\circ \}\)
\begin{align}
     \mathbb{P} \Big ( \limsup_{i \to \infty} \bmu_{k,i}(\theta) = 1 \Big ) = 0.
\end{align}
\end{theorem}
\begin{proof}
See Appendix~\ref{ap:impossible_mislearning}.
\end{proof}

Notice that there is no assumption on the transmission probabilities in Theorem \ref{th:impossible_mislearning}. With bootstrapping, agents never learn a wrong hypothesis. In \cite{bordignon2020partial}, it was shown that agents might mislearn a wrong hypothesis if the fixed transmitted hypothesis is not the true hypothesis --- recall Proposition~\ref{prop:part_information_mislearning}. As a matter of fact, bootstrapping leads to a more robust design in the face of partial communication.

\begin{figure*}[ht]
\centerline{\includegraphics[width=\linewidth]{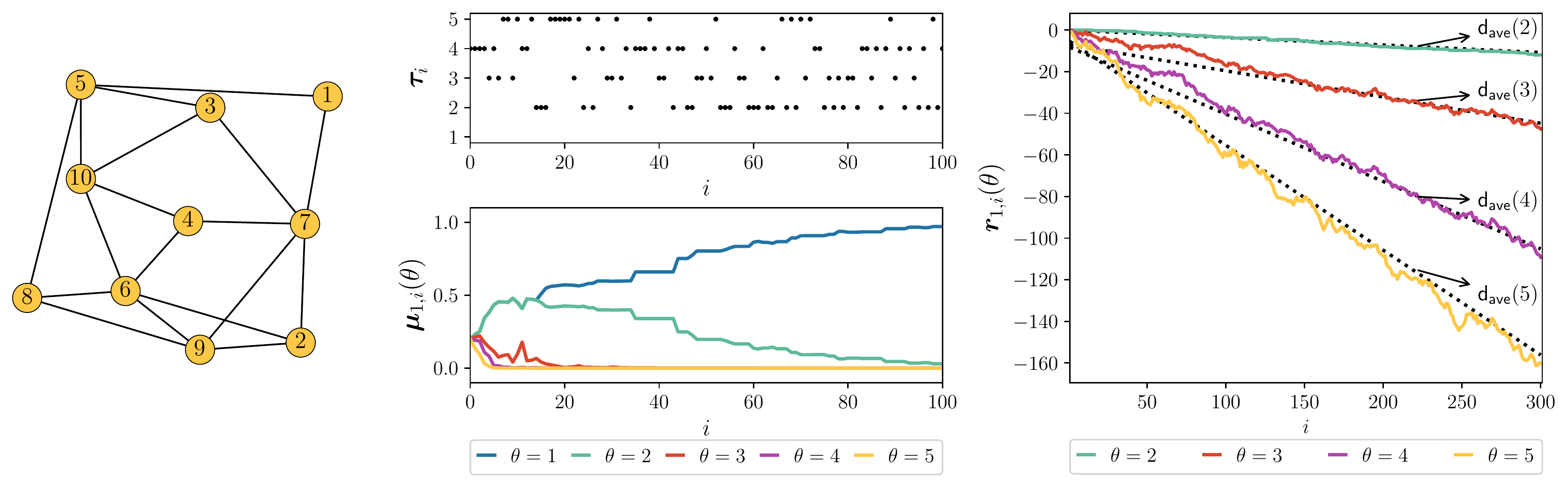}}
\caption{{\em Leftmost panel:} Network topology. {\em Middle panel}: Evolution of the shared hypothesis $\bm{\tau}_i$ over time in the upper panel, and belief evolution for agent 1 showing \emph{truth learning} in the bottom panel. {\em Rightmost panel:} Experimental rates of convergence for agent 1, i.e., $\bm{r}_{1,i}(\theta)$, defined in \eqref{eq:r1i_def}, for different hypotheses (in colored lines), compared with the theoretical asymptotic rate of convergence, i.e., ${\sf d}_{\sf ave}$ defined in \eqref{eq:dave_definition} (in black dotted lines).}
\label{fig:1}
\end{figure*}

\section{Numerical Simulations}\label{sec:simulations}
Consider a $10-$agent strongly-connected network (see the topology in the leftmost panel of Fig.~\ref{fig:1}). The combination matrix is designed using the Metropolis rule~\cite{sayed_2014}, yielding a doubly-stochastic matrix. Agents are trying to detect the true state $\theta^\circ$ among a set of five hypotheses, namely $\Theta\triangleq \{1,2,3,4,5\}$. Incidentally, we assume that $\theta^\circ=1$. To accomplish this task, agents use the protocol described in \eqref{eq:alg_adapt}-\eqref{eq:alg_combine}, where the random shared hypothesis $\bm{\tau}_i$ is distributed according to the following probability mass function:
\begin{equation}
    \mathbb{P}(\boldsymbol{\tau}_i = \theta)=\pi_{\theta}=
    \begin{cases}
    0,&\text{ if }\theta=\theta^\circ\\
    0.25,&\text{ otherwise}.
    \end{cases}
\end{equation}
Agents consider a family of unit-variance Gaussian densities:
\begin{equation}
    f_n(x) = \frac{1}{\sqrt{2\pi}}\exp\Big\{-\frac{(x-0.3n)^2}{2}\Big\}
\end{equation}
for $n=1,2,3,4,5$. The likelihoods of agents are chosen among these Gaussian densities according to the identifiability setup in Table~\ref{tab:2}. For example, we note that agents \(8\)--\(10\) cannot distinguish hypotheses $1$ and $5$. Observe that the global identifiability condition in Assumption \ref{as:global_iden} is satisfied.
\renewcommand{\arraystretch}{1.2}
\begin{table}[ht]
\caption{Identifiability Setup for Network in Fig.~\ref{fig:1}}
\begin{center}
\begin{tabular}{|c|c|c|c|c|c|}
\hline
\multirow{2}{*}{\textbf{Agent} $k$}&\multicolumn{5}{|c|}{\textbf{Likelihood Function:} $L_k(x|\theta)$} \\
\cline{2-6} 
& $\theta=1$& $\theta=2$& $\theta=3$& $\theta=4$& $\theta=5$\\
\hline
$1-2$& $f_1$ &$f_1$ &$f_3$ &$f_4$ &$f_5$\\
$3-5$& $f_1$ &$f_2$ &$f_1$ &$f_4$ &$f_5$ \\
$6-7$& $f_1$ &$f_2$ &$f_3$ &$f_1$ &$f_5$ \\
$8-10$&  $f_1$ &$f_2$ &$f_3$ &$f_4$ &$f_1$\\
\hline
\end{tabular}
\label{tab:2}\end{center}
\end{table}

In the middle panel of Fig.~\ref{fig:1}, we see the evolution of the belief for agent 1, which shows that, although the agents never share information about the true hypothesis, i.e., $\pi_{\theta^\circ}=0$, the agent asymptotically learns the truth, as suggested by Corollary \ref{corollary:truth_learning}. A similar behavior happens for the remaining agents. 

The rightmost panel of Fig.~\ref{fig:1} shows that the experimental convergence rates for agent $1$, i.e.,
\begin{align}\label{eq:r1i_def}
   \bm{r}_{1,i}(\theta) \triangleq   \log \frac{ \bmu_{1,i}(\theta)}{ \bmu_{1,i}(\theta^\circ)} 
\end{align}
which are shown in colored lines, approach the asymptotic convergence rates of traditional social learning (black dotted lines):
\begin{align}\label{eq:dave_definition}
    {\sf d}_{\sf ave} (\theta) \triangleq  \sum_{k = 1}^K  -v_k \dkl\Big (  L_k (\cdot | \theta^\circ) \big |\big| L_k (\cdot | \theta) \Big )
\end{align}
as predicted by Theorem \ref{theorem:asymptotic_rate}. This means that regarding the asymptotic convergence rate, there is no performance loss when only one hypothesis is exchanged at each iteration as long as all wrong hypotheses have positive probability of being transmitted.

\begin{figure}[ht]
\centerline{\includegraphics[width=\linewidth]{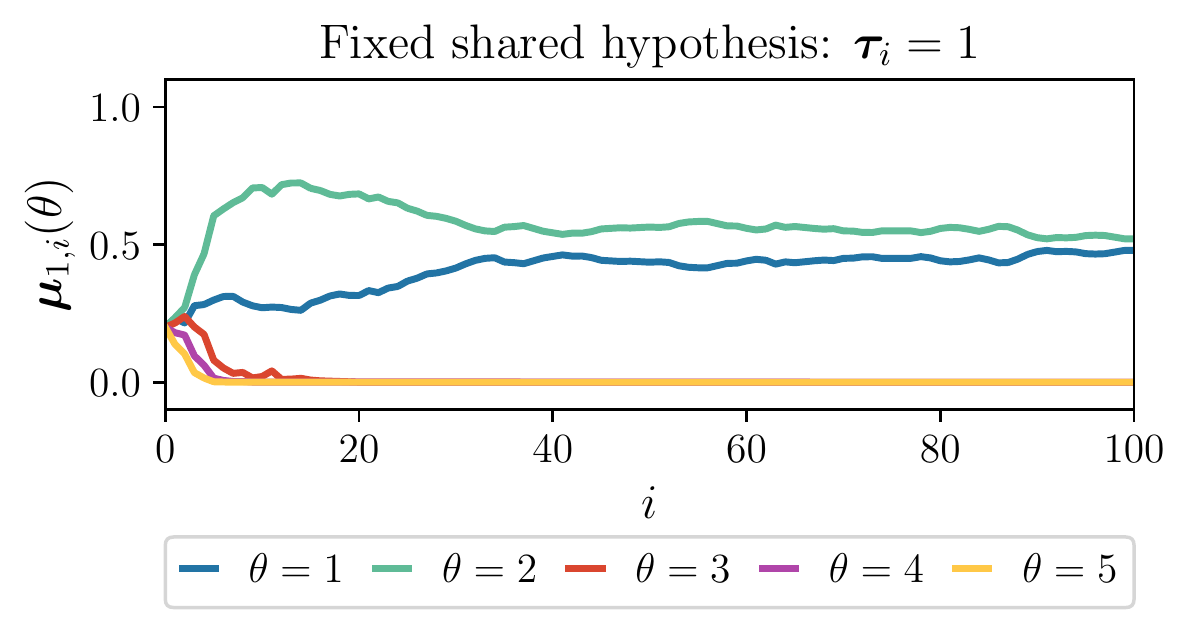}}
\caption{Belief evolution of agent $1$ when the shared hypothesis is \emph{fixed over time} to be the true state of nature. This demonstrates that while there is no truth learning with full confidence, there is no complete rejection of the true hypothesis as well.}
\label{fig:2}
\end{figure}

In the next simulation, we illustrate that truth sharing is not sufficient for truth learning. For that purpose, we fix the transmitted hypothesis at the true hypothesis $\bm{\tau}_i=\theta^\circ=1$ for all $i=1,2,\dots$. The result can be seen in Fig.~\ref{fig:2}, where we show the evolution of the belief of agent $1$ over time. Despite sharing the true hypothesis, the \emph{conservative} behavior described in Section \ref{sec:truth_sharing} hinders the ability of the agent to learn the truth. We note that agent $1$ cannot \emph{decidedly} distinguish between hypothesis $1$ and $2$, which are indistinguishable from its local point of view (see Table~\ref{tab:2}). This was suggested by the example in Section \ref{sec:truth_sharing}, where agents are caught in an equilibrium where they have non-zero belief values for locally indistinguishable hypotheses. Notice that although truth learning is not observed, there is no total mislearning phenomena as well. As suggested by Theorem \ref{th:impossible_mislearning}, the confidence in the truth is not going to 0.

\section{Conclusion}

In this work, we studied social learning under partial information sharing where the transmitted hypothesis is changing at every iteration. In Sec.~\ref{sec:algorithm_description}, we proposed an algorithm for this setting in which agents fill the latent belief components with their own beliefs. In Sec.~\ref{sec:truth_learning}, we derived the rate of convergence under the proposed algorithm and provided sufficient conditions for truth learning. Then, in Secs.~\ref{sec:truth_sharing} and \ref{sec:impossible_mislearning}, we demonstrated that by exchanging beliefs exclusively on the true hypotheses, agents will neither learn the truth with full confidence nor mislearn, i.e., learn a false hypothesis with full confidence. Instead they will be unsure about the truth among their indistinguishable hypotheses. 

There are many possible extensions to the setting considered in this work. For instance, each agent $k$ can choose a possibly different hypothesis $\btau_{k,i}$ to transmit at each iteration.  Alternatively, the transmitted hypothesis $\btau_i$ can evolve according to some Markovian model, instead of independently over time. These extensions introduce additional complexities to the technical analysis, especially when managing the random matrix products. Another interesting extension would be to see if the current work on finite hypothesis sets can extend to continuous hypothesis sets, such as compact sets as in \cite{uribe2022nonasymptotic} or non-compact sets as in the distributed estimation literature (see e.g., \cite{sayed_proc2014,kar2012distributed,dedecius_2017}).

\appendices

\allowdisplaybreaks

\section{Proof Sketch of Proposition~\ref{prop:full_com_learning}}\label{ap:ful_com_learning}
The full proof appears in \cite{nedic_2017,lalitha_2018,kayaalp2022_aa_ga}. The algorithm \eqref{eq:local_bayesian}--\eqref{eq:trad_soc_combine} gives rise to the following recursion for the log-belief ratios:
\begin{equation}\label{eq:log_ratio_trad}
 \log\frac{\bmu_{k,i}(\theta)}{\bmu_{k,i}(\theta^\circ)} = \sum_{\ell \in \mathcal{N}_k} a_{\ell k} \Big ( \log \frac{L_\ell(\bxi_{\ell,i} | \theta)}{L_\ell(\bxi_{\ell,i} | \theta^\circ)} + \log\frac{\bmu_{\ell,i-1}(\theta)}{\bmu_{\ell,i-1}(\theta^\circ)} \Big ).
\end{equation}
Iterating over $i$ and dividing by $i$ gives
\begin{align}\label{eq:ga_i_fold}
 \frac 1 i \log\frac{\bmu_{k,i}(\theta)}{\bmu_{k,i}(\theta^\circ)} &= \frac 1 i  \sum_{j=1}^{i} \sum_{\ell =1}^K [A^{i-j+1}]_{\ell k} \log \frac{L_\ell(\bxi_{\ell,j} | \theta)}{L_\ell(\bxi_{\ell,j} | \theta^\circ)} \notag \\& \qquad + \frac 1 i \sum_{\ell =1}^K [A^{i}]_{\ell k} \log\frac{\bmu_{\ell,0}(\theta)}{\bmu_{\ell,0}(\theta^\circ)} .
\end{align}
Now, since $A$ is primitive and left-stochastic, it holds that $A^i~\to~v\mathds{1}_{K}^T$ \cite[Chapter 8]{horn2012matrix}, as $i \to \infty$. If we incorporate this fact into \eqref{eq:ga_i_fold}, the second term on the right hand side (RHS) vanishes due to Assumption~\ref{as:positive_initial_beliefs} and the expression transforms into:
\begin{align}\label{eq:before_sloln}
 \frac 1 i \log\frac{\bmu_{k,i}(\theta)}{\bmu_{k,i}(\theta^\circ)} &\to \frac 1 i  \sum_{j=1}^{i} \sum_{\ell =1}^K v_{\ell} \log \frac{L_\ell(\bxi_{\ell,j} | \theta)}{L_\ell(\bxi_{\ell,j} | \theta^\circ)} 
\end{align}
as $i$ gets larger. Remember that the observations $\{\bxi_{\ell,i}\}$ are i.i.d. over time, and the expectation of the log-likelihood ratio satisfies:
\begin{equation}\label{eq:expected_log_likelihood}
\e \Big [\log \frac{L_k(\bxi_{k,i} | \theta^\circ)}{L_k(\bxi_{k,i} | \theta)} \Big ] =  \dkl (L_k(.|\theta^{\circ}) || L_k(.|\theta)) < \infty.
\end{equation}
Using \eqref{eq:expected_log_likelihood} and applying the strong law of large numbers \cite[Chapter 7]{williams_1991} to \eqref{eq:before_sloln}, we establish \eqref{eq:ga_standard_rate}. In addition, if Assumption~\ref{as:global_iden} holds, then the RHS of \eqref{eq:ga_standard_rate} is strictly positive for each $\theta \neq \theta^\circ$. Consequently, $\forall \theta \in \Theta \setminus \{ \theta^\circ \}$, $\bmu_{k,i} (\theta) \asc 0 $, which in turn implies \eqref{eq:ga_standard_truth_learning}.

\section{Connection to Stochastic Mirror Descent}\label{ap:mirror_descent}

The goal of social learning can be formulated as solving:
\begin{equation}\label{eq:social_cost}
\min_{\theta \in \Theta} \sum_{k=1}^K \dkl \big (L_k( \cdot |\theta^\circ) || L_k( \cdot |\theta) \big )
\end{equation}
Under Assumption~\ref{as:global_iden}, the unique solution is the true hypothesis \(\theta^\circ\). However, since the true hypothesis is unknown, agents can use stochastic algorithms based on instantaneous data. In \cite{nedic2016}, it is shown that the traditional full information sharing algorithm \eqref{eq:local_bayesian}--\eqref{eq:trad_soc_combine} corresponds to a distributed stochastic mirror descent algorithm for this problem (see \cite{sayed_2023}), with the KL-divergence chosen as the Bregman divergence. Namely, they show that at each iteration, the agents' updates coincide with the distributed stochastic mirror descent algorithm applied to \eqref{eq:social_cost}:
\begin{equation}\label{eq:fusion_rule_objective}
\min_{\mu\in\Delta_H}\!\!\left\{\sum_{\ell\in\mathcal{N}_k}\!\!a_{\ell k}\Big(\dkl(\mu||\bmu_{\ell,i-1})- \,\e_{\mu}\log L_\ell(\bxi_{\ell,i}|\theta)\Big)\!\!\right\}
\end{equation}
where $\Delta_H$ is the probability simplex of dimension $H$, and $\e_{\mu}$ is the expectation computed with respect to $\mu$, i.e.,
\begin{align}
    \e_{\mu} \log L_\ell(\bxi_{\ell,i}|\theta) \triangleq \sum_{\theta \in \Theta} \mu (\theta) \log L_\ell(\bxi_{\ell,i}|\theta).
\end{align}
The objective function in \eqref{eq:fusion_rule_objective} consists of two terms. The first term is the weighted average of KL-divergences across neighboring agents; this average penalizes the disagreement among the neighbors' prior beliefs \( \{ \bmu_{\ell,i-1} \} \). The second term in \eqref{eq:fusion_rule_objective} corresponds to the likelihood of the observations received within the neighborhood, averaged over the hypotheses with respect to \( \mu \). The cost in \eqref{eq:fusion_rule_objective} then seeks to minimize disagreement while maximizing the likelihood.

However, in our problem setting, agents get only the \(\btau_i\) component of their neighbors' beliefs at \(i\). Therefore, problem \eqref{eq:fusion_rule_objective} can be modified to use only that component and fill the rest with the information from agent $k$. That is to say, the KL-divergence is modified as:
\begin{align}
\dkl&(\mu||\bmu_{\ell,i-1}) \Longrightarrow \notag \\
  &\sum_{\theta \neq \btau_i}  \mu (\theta) \log \frac{\mu(\theta)}{\bmu_{k,i-1} (\theta)} + \mu (\btau_i) \log \frac{\mu(\btau_i)}{\bmu_{\ell,i-1} (\btau_i)},
  \end{align}
  and the expectation with respect to \(\mu\) is modified to:
  \begin{align}
\e_{\mu}&\log L_\ell(\bxi_{\ell,i}|\theta)  \Longrightarrow \notag \\ &\sum_{\theta \neq \btau_i} \mu (\theta) \log L_k(\bxi_{k,i}|\theta) +  \mu (\btau_i) \log L_\ell(\bxi_{\ell,i}|\btau_i).
\end{align}

\section{Proof of Lemma~\ref{lemma:submartingale}}\label{ap:submartingale}
From the combination step \eqref{eq:alg_combine},
\begin{align}\label{eq:log_bound_1}
&\log \bmu_{k,i}(\theta^\circ) \notag \\ &\!\!\!=  \sum_{\ell \in \mathcal{N}_k } a_{\ell k} \log \widehat{\bpsi}^{(k)}_{\ell,i} (\theta^\circ ) - \log  \sum_{\theta^\prime \in \Theta}\text{exp} \Big \{ \sum_{\ell \in \mathcal{N}_k } a_{\ell k} \log \widehat{\bpsi}^{(k)}_{\ell,i} (\theta^\prime ) \Big \} \notag \\
&\stackrel{(a)}{\geq} \sum_{\ell \in \mathcal{N}_k } a_{\ell k} \log \widehat{\bpsi}^{(k)}_{\ell,i} (\theta^\circ ) - \log  \sum_{\theta^\prime \in \Theta} \sum_{\ell \in \mathcal{N}_k } a_{\ell k}  \widehat{\bpsi}^{(k)}_{\ell,i} (\theta^\prime ) \notag \\
&= \sum_{\ell \in \mathcal{N}_k } a_{\ell k} \log \widehat{\bpsi}^{(k)}_{\ell,i} (\theta^\circ ) - \log  \sum_{\ell \in \mathcal{N}_k } a_{\ell k}   \sum_{\theta^\prime \in \Theta} \widehat{\bpsi}^{(k)}_{\ell,i} (\theta^\prime ) \notag \\
&\stackrel{(b)}{=} \sum_{\ell \in \mathcal{N}_k } a_{\ell k} \log \widehat{\bpsi}^{(k)}_{\ell,i} (\theta^\circ )
\end{align}
where \( (a) \) follows from Jensen's inequality, and \( (b) \) follows from the fact that $\widehat{\bpsi}^{(k)}_{\ell,i}$ is a pmf and $A$ is left-stochastic. Therefore, the conditional expectations satisfy
\begin{align}\label{eq:log_bound_2}
    &\e \Big [\log \bmu_{k,i} (\theta^\circ)  \Big | \bmf_{i-1} \Big ] \notag \\  &\qquad \geq \e \Big [ \sum_{\ell \in \mathcal{N}_k } a_{\ell k} \log \widehat{\bpsi}^{(k)}_{\ell,i} (\theta^\circ ) \Big | \bmf_{i-1} \Big ] \notag \\
     & \qquad \stackrel{(a)}{=} \e_{\xi_i,\tau_i}  \Big [ \sum_{\ell \in \mathcal{N}_k } a_{\ell k} \log \widehat{\bpsi}^{(k)}_{\ell,i} (\theta^\circ ) \Big ] \notag \\
     & \qquad \stackrel{(b)}{=} \e_{\xi_i} \e_{\tau_i | \xi_i}  \Big [ \sum_{\ell \in \mathcal{N}_k } a_{\ell k} \log \widehat{\bpsi}^{(k)}_{\ell,i} (\theta^\circ ) \Big ] \notag \\
     & \qquad \stackrel{(c)}{=} \e_{\xi_i} \e_{\tau_i}  \Big [ \sum_{\ell \in \mathcal{N}_k } a_{\ell k} \log \widehat{\bpsi}^{(k)}_{\ell,i} (\theta^\circ ) \Big ] 
     \end{align}
where \( (a) \) follows from the fact that arguments inside the expectation are functions of \(\{\bxi_i,\btau_i,\bmf_{i-1}\}\), \( (b) \) follows from the tower rule of expectation, and \( (c) \) follows from the fact that \( \bxi_i\) and \( \btau_i\) are assumed to be independent. The inner expectation can be written as
    \begin{align}\label{eq:log_bound_3}
    &\e_{\tau_i}  \Big [ \sum_{\ell \in \mathcal{N}_k } a_{\ell k} \log \widehat{\bpsi}^{(k)}_{\ell,i} (\theta^\circ )  \Big ] \notag \\ &=\pi_{\theta^\circ}  \sum_{\ell \in \mathcal{N}_k }  a_{\ell k} \log \frac{\bpsi_{\ell,i}(\theta^\circ)}{1-\bpsi_{k,i}(\theta^\circ)+\bpsi_{\ell,i}(\theta^\circ) } \notag \\
    & \qquad + \sum_{\tau \neq \theta^\circ} \pi_{\tau}  \sum_{\ell \in \mathcal{N}_k }  a_{\ell k} \log \frac{\bpsi_{k,i}(\theta^\circ)}{1-\bpsi_{k,i}(\tau)+\bpsi_{\ell,i}(\tau) } \notag \\
    &= \pi_{\theta^\circ}  \sum_{\ell \in \mathcal{N}_k }  a_{\ell k} \log \bpsi_{\ell,i}(\theta^\circ) +  \sum_{\tau \neq \theta^\circ} \pi_{\tau} \log \bpsi_{k,i}(\theta^\circ) \notag \\
    & \qquad - \sum_{\tau \in \Theta} \pi_{\tau} \sum_{\ell \in \mathcal{N}_k }  a_{\ell k} \log \Big (1-\bpsi_{k,i}(\tau)+\bpsi_{\ell,i}(\tau) \Big ).
\end{align}
Using the Perron vector defined by \eqref{eq:perron_def}, taking the expectation of \eqref{eq:log_bound_1} with respect to \( \btau_i\), and using \eqref{eq:log_bound_3} imply that:
\begin{align}\label{eq:inner_exp}
   &\e_{\tau_i}  \Big [ \sum_{k=1}^K v_k \log \bmu_{k,i}(\theta^\circ) \Big ] \notag \\ &\geq \pi_{\theta^\circ}  \sum_{k=1}^K  v_{ k} \log \bpsi_{k,i}(\theta^\circ) +  \sum_{\tau \neq \theta^\circ} \pi_{\tau} \sum_{k=1}^K  v_{ k} \log \bpsi_{k,i}(\theta^\circ) \notag \\
    & \qquad - \sum_{\tau \in \Theta} \pi_{\tau} \sum_{k=1}^K  v_{ k} \sum_{\ell \in \mathcal{N}_k }  a_{\ell k} \log \Big (1-\bpsi_{k,i}(\tau)+\bpsi_{\ell,i}(\tau) \Big ) \notag \\
    &=  \sum_{k=1}^K  v_{ k} \log \bpsi_{k,i}(\theta^\circ) \notag \\
    & \qquad - \sum_{\tau \in \Theta} \pi_{\tau} \sum_{k=1}^K  v_{ k} \sum_{\ell \in \mathcal{N}_k }  a_{\ell k} \log \Big (1-\bpsi_{k,i}(\tau)+\bpsi_{\ell,i}(\tau) \Big ) \notag \\
    & \stackrel{(a)}{\geq} \sum_{k=1}^K  v_{ k} \log \bpsi_{k,i}(\theta^\circ) \notag \\
    & \quad - \log \Bigg ( \sum_{\tau \in \Theta} \pi_{\tau} \sum_{k=1}^K  v_{ k} \sum_{\ell \in \mathcal{N}_k }  a_{\ell k} \big ( 1-\bpsi_{k,i}(\tau)+\bpsi_{\ell,i}(\tau) \big ) \! \Bigg ) \notag \\
    & = \sum_{k=1}^K  v_{ k} \log \bpsi_{k,i}(\theta^\circ) \notag \\
    & \qquad - \log \Bigg ( \sum_{\tau \in \Theta} \pi_{\tau} \sum_{k=1}^K  v_{ k} \big ( 1-\bpsi_{k,i}(\tau)+\bpsi_{k,i}(\tau) \big ) \Bigg ) \notag \\
    & = \sum_{k=1}^K  v_{ k} \log \bpsi_{k,i}(\theta^\circ) 
    - \log  ( 1 ) \notag \\
    & = \sum_{k=1}^K  v_{ k} \log \bpsi_{k,i}(\theta^\circ) ,
\end{align}
where \( (a) \) follows from Jensen's inequality. Applying the expectation with respect to \(\bxi_i\) to the both sides of \eqref{eq:inner_exp}, we arrive at
\begin{align}
    \e_{\xi_i} \e_{\tau_i}&  \Big [ \sum_{k=1}^K v_k \log \bmu_{k,i}(\theta^\circ) \Big ]\notag \\ &\stackrel{\eqref{eq:inner_exp}}{\geq}  \e_{\xi_i} \Big [ \sum_{k=1}^K  v_{ k} \log \bpsi_{k,i}(\theta^\circ) \Big ] \notag \\
    &\stackrel{(a)}{=} \e_{\xi_i} \Big [ \sum_{k=1}^K  v_{ k}  \log \frac{L_k (\bxi_{k,i} | \theta^\circ) }{\sum_{\theta^\prime \in \Theta}L_k (\bxi_{k,i} | \theta^\prime ) \bmu_{k,i-1}(\theta^\prime)}   \Big ] \notag \\ &\qquad+ \sum_{k=1}^K  v_{ k} \log \bmu_{k,i-1}(\theta^\circ) \notag \\
    &=  \sum_{k=1}^K  v_{ k}  \dkl\Big (  L_k (\cdot | \theta^\circ) \Big |\Big| {\sum_{\theta^\prime \in \Theta}L_k (\cdot | \theta^\prime ) \bmu_{k,i-1}(\theta^\prime)} \Big )  \notag \\ & \qquad+ \sum_{k=1}^K  v_{ k} \log \bmu_{k,i-1}(\theta^\circ) \notag \\
    &\stackrel{(b)}{\geq} \sum_{k=1}^K  v_{ k} \log \bmu_{k,i-1}(\theta^\circ) \label{eq:martingale_data}
\end{align}
where \( (a) \) follows from the fact that belief vectors at time $i-1$ are independent of the new observations at time $i$, and \( (b) \) follows from the fact that KL-divergences are non-negative. It is worth noting that this proof holds even in cases where the transmission distribution is time-dependent, or when it depends on the observations. Therefore, the result is more general than what is explicitly stated. However, to maintain consistency with the other parts of the paper, we used the fixed distribution notation \(\pi\).

\section{Proof of Theorem~\ref{theorem:asymptotic_rate}}\label{ap:rate_convergence}
By the proposed strategy \eqref{eq:alg_adapt}--\eqref{eq:alg_combine}, the log-belief ratio can be written as
\begin{align}\label{eq:logbelief_equation}
    \log \frac{ \bmu_{k,i}(\theta)}{ \bmu_{k,i}(\theta^\circ)}  &= \sum_{\ell \in \mathcal{N}_k } a_{\ell k} \log \frac{\widehat{\bpsi}^{(k)}_{\ell,i} (\theta )}{\widehat{\bpsi}^{(k)}_{\ell,i} (\theta^\circ)}  \notag \\ &= \mathbb{I} \{ \btau_i = \theta \} \sum_{\ell \in \mathcal{N}_k } a_{\ell k}  \log \frac{\bpsi_{\ell,i}(\theta)}{\bpsi_{k,i}(\theta^\circ)} \notag \\&\quad + \mathbb{I} \{ \btau_i = \theta^\circ \} \sum_{\ell \in \mathcal{N}_k } a_{\ell k}  \log \frac{\bpsi_{k,i}(\theta)}{\bpsi_{\ell,i}(\theta^\circ)} \notag \\ & \quad + (1-\mathbb{I} \{ \btau_i = \theta \}- \mathbb{I} \{ \btau_i = \theta^\circ \}) \log \frac{\bpsi_{k,i}(\theta)}{\bpsi_{k,i}(\theta^\circ)}.
\end{align}
Observe from \eqref{eq:logbelief_equation} that the log-belief ratio is a random variable given the intermediate beliefs, because of the randomness of the trending topic \( \btau_i\). Next, we fix a wrong hypothesis \( \theta \neq \theta^\circ \) and define the effective combination matrix for \( \theta \) at time \( i \) as
\begin{align}\label{eq:ai_binary_rv}
    \widetilde{\boldsymbol{A}}_i  \triangleq \begin{dcases} 
     A , & \theta = \btau_i \\
     I , & \theta \neq \btau_i
   \end{dcases}.
\end{align}
This is a binary random variable taking the value of the original combination matrix \( A \) if the hypothesis is exchanged at iteration \( i \), and the identity matrix \( I \) otherwise. More compactly, 
\begin{align}
    \widetilde{\boldsymbol{A}}_i = \mathbb{I} \{ \btau_i = \theta \} A + (1-\mathbb{I} \{ \btau_i = \theta \}) I.
\end{align}
Using this definition, the recursion in \eqref{eq:logbelief_equation} can be rewritten for $\theta\neq\theta^\circ$ as:
\begin{align}\label{eq:one_step_recursion}
    &\log \frac{ \bmu_{k,i}(\theta)}{ \bmu_{k,i}(\theta^\circ)} \notag \\ &= \sum_{\ell \in \mathcal{N}_k }   [\widetilde{\boldsymbol{A}}_i  ]_{\ell k} \log \frac{\bpsi_{\ell,i}(\theta)}{\bpsi_{\ell,i}(\theta^\circ)} \notag \\ &\qquad + (\mathbb{I} \{ \btau_i = \theta^\circ \}-\mathbb{I} \{ \btau_i = \theta \})\sum_{\ell \in \mathcal{N}_k } a_{\ell k} \log \frac{\bpsi_{k,i}(\theta^\circ)}{\bpsi_{\ell,i}(\theta^\circ)} \notag \\
    &\stackrel{\eqref{eq:alg_adapt}}{=}\sum_{\ell \in \mathcal{N}_k }   [\widetilde{\boldsymbol{A}}_i  ]_{\ell k} \log \frac{L_\ell (\bxi_{\ell,i} | \theta )}{L_\ell (\bxi_{\ell,i} | \theta^\circ)} + \sum_{\ell \in \mathcal{N}_k }   [\widetilde{\boldsymbol{A}}_i  ]_{\ell k} \log \frac{\bmu_{\ell,i-1}(\theta)}{\bmu_{\ell,i-1}(\theta^\circ)} \notag \\ & \quad  + (\mathbb{I} \{ \btau_i = \theta^\circ \}-\mathbb{I} \{ \btau_i = \theta \})\sum_{\ell \in \mathcal{N}_k } a_{\ell k} \log \frac{\bpsi_{k,i}(\theta^\circ)}{\bpsi_{\ell,i}(\theta^\circ)} .
\end{align}
The first two terms in the RHS of \eqref{eq:one_step_recursion} are analogous to the terms that arise in the standard log-linear social learning analysis (see Eq.~\eqref{eq:log_ratio_trad}), albeit with the random matrix \( \widetilde{\boldsymbol{A}}_i  \) in place of the original combination matrix \( A \). The last term in \eqref{eq:one_step_recursion} is a residue term due to the network disagreement. In order to expand the recursion over time, we introduce the following notation for the product of the effective combination matrices for \( j \leq i \):\footnote{If $j> i$, we set $\widetilde{\boldsymbol{A}}^{j \shortrightarrow i} = I$.}
\begin{align}
    \widetilde{\boldsymbol{A}}^{j \shortrightarrow i} \triangleq \widetilde{\boldsymbol{A}}_j\widetilde{\boldsymbol{A}}_{j+1}\dots \widetilde{\boldsymbol{A}}_{i-1}\widetilde{\boldsymbol{A}}_i,
\end{align}
and also for the residue terms for \( j~\leq~i \) as
\begin{align}\label{eq:expanded_rjk}
    \boldsymbol{R}_j^k \triangleq  &(\mathbb{I} \{ \btau_j = \theta^\circ \}-\mathbb{I} \{ \btau_j = \theta \}) \notag \\& \qquad \times \sum_{\ell = 1}^K  [\widetilde{\boldsymbol{A}}^{j+1 \shortrightarrow i}]_{\ell k} \sum_{m=1}^K a_{m\ell} \log \frac{\bpsi_{\ell,j}(\theta^\circ)}{\bpsi_{m,j}(\theta^\circ)}.
\end{align}
In other words, $\boldsymbol{R}_j^k$ denotes the residual term at time $i$ caused by the network disagreement at time $j$. Expanding \eqref{eq:one_step_recursion} with these definitions and dividing both sides by $i$, we arrive at the following expression for the convergence rate.
\begin{align}\label{eq:expanded_logbelief}
    \frac 1 i \log \frac{ \bmu_{k,i}(\theta)}{ \bmu_{k,i}(\theta^\circ)} &= \frac 1 i \sum_{j=1}^i \sum_{\ell = 1}^K  [\widetilde{\boldsymbol{A}}^{j \shortrightarrow i}]_{\ell k} \log \frac{L_\ell (\bxi_{\ell,j} | \theta )}{L_\ell (\bxi_{\ell,j} | \theta^\circ)} \notag \\
    & \quad  + \frac 1 i \sum_{\ell = 1}^K  [\widetilde{\boldsymbol{A}}^{1 \shortrightarrow i}]_{\ell k} \log \frac{\bmu_{\ell,0}(\theta)}{\bmu_{\ell,0}(\theta^\circ)} + \frac 1 i \sum_{j=1}^i \boldsymbol{R}_j^k.
\end{align}
In Lemma~\ref{lemma:residual}, we show that the summation of the residue terms in \eqref{eq:expanded_logbelief} stays finite with probability one as $i$ grows, i.e.,
\begin{align}
    \frac{1}{i} \sum_{j=1}^i \boldsymbol{R}_j^k \asc 0.
\end{align}
Therefore, the residue terms do not affect the convergence rate in \eqref{eq:expanded_logbelief} as $i \to \infty$. Consequently, we proceed to study the remaining terms in \eqref{eq:expanded_logbelief}. 
Observe that the finiteness of the KL-divergence of likelihood functions (see Eq.~\eqref{eq:assumption_finite_kl}) can be expressed as
\begin{align}
    \Big | \mathbb{E}\Big[\log \frac{L_k (\bxi_{k,i} | \theta )}{L_k (\bxi_{k,i} | \theta^\circ)} \Big ] \Big | < \infty 
\end{align}
which in turn implies:
\begin{align}\label{eq:likelihood_as}
    \Big | \log \frac{L_k (\bxi_{k,i} | \theta )}{L_k (\bxi_{k,i} | \theta^\circ)} \Big | \stackrel{\text{a.s.}}{<} \infty.
\end{align}
The set of events such that
\begin{align}
 \lim_{n \to \infty } [\widetilde{\boldsymbol{A}}^{j \shortrightarrow j+n}-v\mathds{1}_K^\top]_{\ell k} \log \frac{L_\ell (\bxi_{\ell,j} | \theta )}{L_\ell (\bxi_{\ell,j} | \theta^\circ)} \neq 0
\end{align}
 is a subset of the union of the event sets:
\begin{align}
    \lim_{n \to \infty } [\widetilde{\boldsymbol{A}}^{j \shortrightarrow j+n}-v\mathds{1}_K^\top]_{\ell k} \neq 0
\end{align}
and
\begin{align}
   \Big | \log \frac{L_\ell (\bxi_{\ell,j} | \theta )}{L_\ell (\bxi_{\ell,j} | \theta^\circ)} \Big | = \infty .
\end{align}
But we know that these two sets are null sets because of the auxiliary Lemma~\ref{lemma:matrix_convergence} and \eqref{eq:likelihood_as}, respectively. Therefore, for any time instant $j$, it holds that
\begin{align}
 \lim_{n \to \infty } [\widetilde{\boldsymbol{A}}^{j \shortrightarrow j+n}-v\mathds{1}_K^\top]_{\ell k} \log \frac{L_\ell (\bxi_{\ell,j} | \theta )}{L_\ell (\bxi_{\ell,j} | \theta^\circ)}  \asceq 0.
\end{align}
This implies for the convergence of the Ces{\`a}ro mean \cite{cesaro1888convergence} that
\begin{align}
 \lim_{t \to \infty } \frac{1}{t} \sum_{n=0}^{t-1} [\widetilde{\boldsymbol{A}}^{j \shortrightarrow j+n}-v\mathds{1}_K^\top]_{\ell k} \log \frac{L_\ell (\bxi_{\ell,j} | \theta )}{L_\ell (\bxi_{\ell,j} | \theta^\circ)}  \asceq 0. 
\end{align}
By using $j=i-n$, this can alternatively be written as
\begin{align}
 \lim_{t \to \infty } \frac{1}{t} \sum_{n=0}^{t-1} [\widetilde{\boldsymbol{A}}^{i-n \shortrightarrow i}-v\mathds{1}_K^\top]_{\ell k} \log \frac{L_\ell (\bxi_{\ell,i-n} | \theta )}{L_\ell (\bxi_{\ell,i-n} | \theta^\circ)}  \asceq 0. 
\end{align}
Since this holds for any time instant $i\geq t$ (which ensures $j\geq 1$), we can set $i=t$. By also changing the summation index $n$ to $j$, we arrive at
\begin{align}\label{eq:cesaro_mean_as}
   \lim_{i \to \infty} \frac{1}{i} \sum_{j=1}^i[\widetilde{\boldsymbol{A}}^{j \shortrightarrow i}-v\mathds{1}_K^\top]_{\ell k} \log \frac{L_\ell (\bxi_{\ell,j} | \theta )}{L_\ell (\bxi_{\ell,j} | \theta^\circ)}  \asceq 0.
\end{align} 
As a result, the first term in \eqref{eq:expanded_logbelief} can be written as:
\begin{align}
    \frac{1}{i} \sum_{j=1}^i& \sum_{\ell = 1}^K  [\widetilde{\boldsymbol{A}}^{j \shortrightarrow i}]_{\ell k} \log \frac{L_\ell (\bxi_{\ell,j} | \theta )}{L_\ell (\bxi_{\ell,j} | \theta^\circ)} \notag \\ &= \frac{1}{i} \sum_{j=1}^i \sum_{\ell = 1}^K  [\widetilde{\boldsymbol{A}}^{j \shortrightarrow i} - v\mathds{1}_K^\top]_{\ell k} \log \frac{L_\ell (\bxi_{\ell,j} | \theta )}{L_\ell (\bxi_{\ell,j} | \theta^\circ)} \notag \\ &\quad + \frac{1}{i} \sum_{j=1}^i \sum_{\ell = 1}^K  [v\mathds{1}_K^\top]_{\ell k} \log \frac{L_\ell (\bxi_{\ell,j} | \theta )}{L_\ell (\bxi_{\ell,j} | \theta^\circ)} \notag \\
    &= \frac{1}{i} \sum_{j=1}^i \sum_{\ell = 1}^K  [\widetilde{\boldsymbol{A}}^{j \shortrightarrow i} - v\mathds{1}_K^\top]_{\ell k} \log \frac{L_\ell (\bxi_{\ell,j} | \theta )}{L_\ell (\bxi_{\ell,j} | \theta^\circ)} \notag \\ &\quad + \frac{1}{i} \sum_{j=1}^i \sum_{k = 1}^K  v_k \log \frac{L_k (\bxi_{k,j} | \theta )}{L_k (\bxi_{k,j} | \theta^\circ)} \notag \\
    &= \frac{1}{i} \sum_{j=1}^i \sum_{\ell = 1}^K  [\widetilde{\boldsymbol{A}}^{j \shortrightarrow i} - v\mathds{1}_K^\top]_{\ell k} \log \frac{L_\ell (\bxi_{\ell,j} | \theta )}{L_\ell (\bxi_{\ell,j} | \theta^\circ)} \notag \\ &\quad + \sum_{k = 1}^K  v_k \frac{1}{i} \sum_{j=1}^i  \log \frac{L_k (\bxi_{k,j} | \theta )}{L_k (\bxi_{k,j} | \theta^\circ)} \notag \\
    &\asc \sum_{k = 1}^K  -v_k \dkl\Big (  L_k (\cdot | \theta^\circ) \Big |\Big| L_k (\cdot | \theta) \Big )
\end{align}
where the last step follows from \eqref{eq:cesaro_mean_as} and the strong law of large numbers \cite[Chapter 7]{williams_1991}. Also, since by Lemma~\ref{lemma:matrix_convergence} \( \widetilde{\boldsymbol{A}}^{1 \shortrightarrow i}  \asc  v\mathds{1}_K^\top \) and by Assumption~\ref{as:positive_initial_beliefs} the initial beliefs are nonzero, it follows that
\begin{align}
    \frac{1}{i} \sum_{\ell = 1}^K  [\widetilde{\boldsymbol{A}}^{1 \shortrightarrow i}]_{\ell k} \log \frac{\bmu_{\ell,0}(\theta)}{\bmu_{\ell,0}(\theta^\circ)} & \asc 0.
\end{align}
The asymptotic convergence rate then becomes
\begin{align}
    \frac{1}{i} \log \frac{ \bmu_{k,i}(\theta)}{ \bmu_{k,i}(\theta^\circ)} \asc \sum_{k = 1}^K  -v_k \dkl\Big (  L_k (\cdot | \theta^\circ) \Big |\Big| L_k (\cdot | \theta) \Big ).
\end{align}
\qed

\section{Proof of Theorem~\ref{th:impossible_mislearning}}\label{ap:impossible_mislearning}
Since \( Q(\bmu_i) \) is a super-martingale (Lemma~\ref{lemma:submartingale}) and also non-negative (i.e., uniformly bounded from below), by Doob's forward martingale convergence theorem \cite[Chapter 11.5]{williams_1991}, there exists a finite random variable \( \bm{Q}_{\infty} \) such that, as \( i \to \infty\),
\begin{align}\label{eq:q_inf_intro}
    Q(\bmu_i) \asc \bm{Q}_{\infty}.
\end{align}
Since \( \bm{Q}_{\infty} \) is finite, it holds that:
\begin{align}
    &\lim_{i \to \infty }\sum_{k=1}^K v_k \log \bmu_{k,i}(\theta^\circ) > -\infty \notag \\
    &\Longrightarrow \liminf_{i \to \infty} \log \bmu_{k,i} (\theta^\circ) > -\infty, \quad \forall k\in\mathcal{N} \notag \\
    &\Longrightarrow \liminf_{i \to \infty}  \bmu_{k,i} (\theta^\circ) > 0 
\end{align}
with probability one.

\section{Auxiliary Results}

\subsection{Vanishing Matrix Norm}

\begin{lemma}\label{lemma:vanishing_norm} If the wrong hypothesis \( \theta\) is transmitted with positive probability, i.e., \( \pi_{\theta} > 0\), then, for any matrix norm induced by a vector norm,
\begin{align}
\e \Big [ \left \| (\widetilde{\boldsymbol{A}}^{j+1 \shortrightarrow i})^\top (I-A^\top) \right \| \Big] = \mathcal{O} (\widetilde{\lambda}^{i-j})
\end{align}
for a constant \( \widetilde{\lambda} \) that satisfies \( 0 \leq \widetilde{\lambda} < 1\).
\end{lemma}
\begin{proof}
Define the time difference \( n \triangleq i-j \). Since the matrices \( \widetilde{\boldsymbol{A}}_i \) are i.i.d. binary random variables over time, as defined by \eqref{eq:ai_binary_rv}, for time differences \( 0 \leq m \leq n \), we get:
\begin{align}\label{eq:binomial_probability}
    \mathbb{P} \Big (\widetilde{\boldsymbol{A}}^{j+1 \shortrightarrow i} = A^m \Big ) = \binom{n}{m} (\pi_{\theta})^m (1-\pi_{\theta})^{n-m}.
\end{align}
Moreover, since \( A \) is a primitive stochastic matrix, for consecutive time instants, there exists a non-negative constant \( \lambda < 1 \) such that \cite[Eq. (8.2.10)]{horn2012matrix}:
\begin{align}\label{eq:eigenvalue_decrease}
    \| A^{m} - A^{m+1} \| \leq C \lambda^{m} (1-\lambda )
\end{align}
where \( C \) is a constant independent of \( m \). Then, it follows that
\begin{align}
    \e \Big [\| & (\widetilde{\boldsymbol{A}}^{j+1 \shortrightarrow i})^\top (I-A^\top)  \| \Big] \notag \\ &= \sum_{m=0}^n \mathbb{P}(\widetilde{\boldsymbol{A}}^{j+1 \shortrightarrow i} = A^m) \| A^{m} - A^{m+1} \|  \notag \\
    &\stackrel{\eqref{eq:eigenvalue_decrease}}{\leq} \sum_{m=0}^n \mathbb{P}(\widetilde{\boldsymbol{A}}^{j+1 \shortrightarrow i} = A^m) C \lambda^{m} (1-\lambda ) \notag \\
    &\stackrel{\eqref{eq:binomial_probability}}{=}C (1-\lambda) \sum_{m=0}^n \binom{n}{m} (\pi_{\theta})^m (1-\pi_{\theta})^{n-m}  \lambda^{m} \notag \\
    &=C (1-\lambda)  \Big (\lambda \pi_{\theta} + (1-\pi_{\theta})  \Big )^n \notag \\
    &=C (1-\lambda) \Big (1- (1-\lambda)\pi_{\theta} \Big )^n = \mathcal{O} (\widetilde{\lambda}^n)
\end{align}
where \( \widetilde{\lambda} \triangleq 1- (1-\lambda)\pi_{\theta} \), which is a constant strictly smaller than \( 1 \) as long as \( \pi_{\theta} > 0 \).
\end{proof}

\subsection{Finiteness of the Residual Sum}\label{ap:residual}
\begin{lemma}~\label{lemma:residual}As \( i \to \infty \), if $\pi_\theta > 0$, then, under Assumption~\ref{as:positive_initial_beliefs}
\begin{align}
\frac{1}{i}\sum_{j=1}^i \boldsymbol{R}_j^k \asc 0.
\end{align}
\end{lemma}
\begin{proof}
First, we aggregate the residue terms and log-intermediate beliefs from across the network into the following vectors:
\begin{align}
    \boldsymbol{R}_j \triangleq \text{col} \{ \boldsymbol{R}_j^k \}_{k=1}^K, \quad 
   \boldsymbol{\Psi}_{i} \triangleq \text{col} \{\log \bpsi_{k,i}(\theta^\circ)\}_{k=1}^K.
\end{align}
By using these definitions, expression \eqref{eq:expanded_rjk} can be transformed into the following vector equation form:
\begin{align}\label{eq:rj_definition}
     \boldsymbol{R}_j =  (\mathbb{I} \{ \btau_j = \theta^\circ \}-\mathbb{I} \{ \btau_j = \theta \})  (\widetilde{\boldsymbol{A}}^{j+1 \shortrightarrow i})^\top  (  \boldsymbol{\Psi}_{j} - A^\top \boldsymbol{\Psi}_{j}).
\end{align}
To bound the terms in this expression, we start by using Lemma~\ref{lemma:vanishing_norm} and Markov's inequality to obtain:
\begin{align}
    \mathbb{P} \Big ( \| ( \widetilde{\boldsymbol{A}}^{j+1 \shortrightarrow j+n})^\top (I-A^\top) \| \geq \epsilon \Big ) \leq \frac{C (1-\lambda) \widetilde{\lambda}^n}{\epsilon},
\end{align}
where recall that we use \( n = i-j \). Since \( \widetilde{\lambda} < 1 \), it holds that
\begin{align}
   \sum_{n=0}^\infty \frac{C (1-\lambda) \widetilde{\lambda}^n}{\epsilon} < \infty.
\end{align}
Then, the first Borel-Cantelli Lemma \cite[Chapter 2.7]{williams_1991} implies:
\begin{align}\label{eq:n_limit_matrix_as}
 \lim_{n \to \infty}  \Big \| (\widetilde{\boldsymbol{A}}^{j+1 \shortrightarrow j+n})^\top (I-A^\top) \Big \| \asceq 0,
\end{align}
 for any value of $j$. Moreover, if we bound the norm of \eqref{eq:rj_definition}, it holds almost surely that
\begin{align}
    & \| \boldsymbol{R}_j \| \notag \\&=  \Big \| (\mathbb{I} \{ \btau_j = \theta^\circ \}-\mathbb{I} \{ \btau_j = \theta \})  (\widetilde{\boldsymbol{A}}^{j+1 \shortrightarrow i})^\top  (  \boldsymbol{\Psi}_{j} - A^\top \boldsymbol{\Psi}_{j}) \Big \| \notag \\
     &\stackrel{(a)}{\leq} \Big \|\mathbb{I} \{ \btau_j = \theta^\circ \}-\mathbb{I} \{ \btau_j = \theta \} \Big \| \Big \|  (\widetilde{\boldsymbol{A}}^{j+1 \shortrightarrow i})^\top (I-A^\top) \Big \| \Big  \| \boldsymbol{\Psi}_{j} \Big \| \notag \\
     &\leq \Big \|  (\widetilde{\boldsymbol{A}}^{j+1 \shortrightarrow i})^\top (I-A^\top) \Big \| \Big  \| \boldsymbol{\Psi}_{j} \Big \| \notag \\
      &\stackrel{(b)}{\leq} \Big \|  (\widetilde{\boldsymbol{A}}^{j+1 \shortrightarrow i})^\top (I-A^\top) \Big \| \boldsymbol{\Psi} \label{eq:rj_bound_psi},
\end{align}
where \( (a) \) follows from the sub-multiplicity of the norm, and \((b) \) follows from the definition
\begin{equation}\label{eq:psi_definition}
   \boldsymbol{\Psi} \triangleq \sup_{j\geq 1} \| \boldsymbol{\Psi}_{j} \| ,
\end{equation}
which is shown to be finite under Assumption~\ref{as:positive_initial_beliefs} in Appendix~\ref{ap:auxiliary_res}. Subsequently, the norm of the Ces{\`a}ro mean satisfies
\begin{align}\label{eq:cesaro_mean_prev}
\Big \| \frac{1}{i}\sum_{j=1}^i \boldsymbol{R}_j \Big \| &\leq  \frac{1}{i}\sum_{j=1}^{i} \Big \|\boldsymbol{R}_j \Big \| \notag \\
&\stackrel{\eqref{eq:rj_bound_psi}}{\leq} \boldsymbol{\Psi} \frac{1}{i}\sum_{j=1}^{i} \Big \|  (\widetilde{\boldsymbol{A}}^{j+1 \shortrightarrow i})^\top (I-A^\top) \Big \| .
\end{align}
Observe that \eqref{eq:n_limit_matrix_as} can alternatively be written as (by using the definition $n=i-j$)
\begin{align}
 \lim_{n \to \infty}  \Big \| (\widetilde{\boldsymbol{A}}^{i-n+1 \shortrightarrow i})^\top (I-A^\top) \Big \| \asceq 0.
\end{align}
 As a result, the Ces{\`a}ro mean satisfies
\begin{align}
    \lim_{t \to \infty} \frac{1}{t}\sum_{n=0}^{t-1} \Big \|  (\widetilde{\boldsymbol{A}}^{i-n+1 \shortrightarrow i})^\top (I-A^\top) \Big \| \asceq 0,
\end{align}
for any $i\geq t$ (so that $j=i-n\geq 1$). If we set $i=t$, and change the indices from $n$ to $j=i-n$, we get
\begin{equation}
\lim_{i \to \infty} \frac{1}{i}\sum_{j=1}^{i} \Big \|  (\widetilde{\boldsymbol{A}}^{j+1 \shortrightarrow i})^\top (I-A^\top) \Big \|
    \asceq 0.
\end{equation}
Incorporating this into \eqref{eq:cesaro_mean_prev}, we conclude that, as $i \to \infty$,
\begin{align}
   \Big \| \frac{1}{i}\sum_{j=1}^i \boldsymbol{R}_j \Big \| \asc 0, \quad \Longrightarrow   \frac{1}{i}\sum_{j=1}^i \boldsymbol{R}_j^k  \asc 0.
\end{align}
\end{proof}

\subsection{Convergence of the Matrix Product}\label{ap:matrix_convergence}
\begin{lemma}~\label{lemma:matrix_convergence}For a fixed time instant $j$, if $\pi_\theta > 0$, then
\begin{align}
   \widetilde{\boldsymbol{A}}^{j \shortrightarrow i} \asc  v\mathds{1}_K^\top
\end{align}
as \( i \to \infty\).
\end{lemma}
\begin{proof}
Recall from \eqref{eq:ai_binary_rv} that \( \widetilde{\boldsymbol{A}}_i \) is a binary random variable and let \( E_i \) denote the event that \(\widetilde{\boldsymbol{A}}_i = A \). Since \( \{ E_i \}_{i=1}^\infty  \) are independent events across time and their probability of occurrence satisfies $\pi_\theta >0$, it holds that
\begin{align}
    \sum_{i=1}^\infty \mathbb{P}(E_i) = \sum_{i=1}^\infty \pi_\theta = \infty .
\end{align}
Subsequently, by the second Borel-Cantelli Lemma \cite[Chapter 4.3]{williams_1991}, we conclude that
\begin{align}\label{eq:ei_occurs_inf}
    \mathbb{P}(E_i \hspace{0.25em} \text{occurs for infinitely many} \hspace{0.25em}  i) = 1.
\end{align}
Notice that in the product of random matrices
\begin{align}
    \widetilde{\boldsymbol{A}}^{j \shortrightarrow i} = \widetilde{\boldsymbol{A}}_j\widetilde{\boldsymbol{A}}_{j+1}\dots \widetilde{\boldsymbol{A}}_{i-1}\widetilde{\boldsymbol{A}}_i,
\end{align}
the realization of random matrices will either be equal to \( A \) or the identity matrix \( I \). Multiplication with identity matrices has no effect on the product and by \eqref{eq:ei_occurs_inf}, there will be infinitely many \( A \)'s in the product with probability one, as \( n = i-j \to \infty \). Thus, if we define $A^\infty \triangleq \lim_{i \to \infty} A^i$,
\begin{align}
    \mathbb{P} (\lim_{n \to \infty} \widetilde{\boldsymbol{A}}^{j \shortrightarrow j+n} = A^\infty) = 1.
\end{align}
Finally, since $A^\infty = v\mathds{1}_K^\top $ due to $A$ being a primitive matrix, we conclude that:
\begin{align}\label{eq:random_matrix_as}
    \lim_{n \to \infty} \widetilde{\boldsymbol{A}}^{j \shortrightarrow j+n} \asceq  v\mathds{1}_K^\top
\end{align}
for any time instant $j\geq 1$.
\end{proof}

\subsection{Uniform Boundedness}\label{ap:auxiliary_res}
In this section, we show that, under Assumption~\ref{as:positive_initial_beliefs}
\begin{equation}\label{eq:psi_finiteness}
   \boldsymbol{\Psi} = \sup_{j\geq 1} \| \boldsymbol{\Psi}_{j} \| \stackrel{\text{a.s.}}{<} \infty .
\end{equation}
For that purpose, first, we show that Lemma~\ref{lemma:submartingale} implies that \( \sum_{k=1}^K  v_{ k} \log \bpsi_{k,i}(\theta^\circ)  \) is a sub-martingale. To see this, observe that by \eqref{eq:martingale_data},
\begin{equation}
    \e_{\xi_i} \Big [ \sum_{k=1}^K  v_{ k} \log \bpsi_{k,i}(\theta^\circ)  \Big ] \geq \sum_{k=1}^K  v_{ k} \log \bmu_{k,i-1}(\theta^\circ) ,
\end{equation}
and in turn, by \eqref{eq:inner_exp},
\begin{equation}
    \e_{\tau_{i-1}} \Big [ \sum_{k=1}^K  v_{k}  \log \bmu_{k,i-1}(\theta^\circ) \Big ] \geq \sum_{k=1}^K  v_{ k} \log \bpsi_{k,i-1}(\theta^\circ).
\end{equation}
Since it is also a non-positive sub-martingale, it converges to a finite limit almost surely \cite[Chapter 11.5]{williams_1991}, which means that as \( i \to \infty \),
\begin{align}\label{eq:norm_as_finite}
    \sum_{k=1}^K  v_{ k} \log \bpsi_{k,i}(\theta^\circ) &\stackrel{\text{a.s.}}{>} -\infty \notag \\
    \Longrightarrow \forall k \in \mathcal{N}, \quad \log \bpsi_{k,i}(\theta^\circ)  &\stackrel{\text{a.s.}}{>} -\infty \notag \\
    \Longrightarrow \sum_{k=1}^K -\log \bpsi_{k,i}(\theta^\circ)  &\stackrel{\text{a.s.}}{<} \infty \notag \\
    \Longrightarrow \| \boldsymbol{\Psi}_{i} \| &\stackrel{\text{a.s.}}{<} \infty.
\end{align}
In addition, for any finite time instant \( j \), it is true that \(  \bpsi_{k,i} (\theta^\circ)> 0 \) for each agent $k$ because of the following reasons. Due to Assumption~\ref{as:positive_initial_beliefs}, the initial beliefs are positive. Moreover, the likelihood at the true hypothesis by definition cannot be zero for emitted observations. Furthermore, the geometric combination rule results in the intersection of the supports of its arguments \cite{koliander2022}. Consequently, for any time instant \( j \), it is true that \(  \| \boldsymbol{\Psi}_{j} \| < \infty \). Combining this with \eqref{eq:norm_as_finite} establishes \eqref{eq:psi_finiteness}.

\bibliographystyle{IEEEbib}
\bibliography{refs}

\begin{thebibliography}{10}

\bibitem{chamley_2003}
C.~P. Chamley,
\newblock {\em Rational Herds: Economic Models of Social Learning},
\newblock Cambridge University Press, 2003.

\bibitem{djuric2012}
P.~M. Djurić and Y.~Wang,
\newblock ``Distributed {Bayesian} learning in multiagent systems: Improving
  our understanding of its capabilities and limitations,''
\newblock {\em IEEE Signal Processing Magazine}, vol. 29, no. 2, pp. 65--76,
  2012.

\bibitem{chamley2013}
C.~Chamley, A.~Scaglione, and L.~Li,
\newblock ``Models for the diffusion of beliefs in social networks: An
  overview,''
\newblock {\em IEEE Signal Processing Magazine}, vol. 30, no. 3, pp. 16--29,
  2013.

\bibitem{krishnamurthy_2013}
V.~Krishnamurthy and H.~V. Poor,
\newblock ``Social learning and {Bayesian} games in multiagent signal
  processing: how do local and global decision makers interact?,''
\newblock {\em IEEE Signal Processing Magazine}, vol. 30, no. 3, pp. 43--57,
  2013.

\bibitem{mossel2017opinion}
E.~Mossel and O.~Tamuz,
\newblock ``Opinion exchange dynamics,''
\newblock {\em Probability Surveys}, vol. 14, pp. 155--204, 2017.

\bibitem{bordignon2021learning}
V.~Bordignon, S.~Vlaski, V.~Matta, and A.~H. Sayed,
\newblock ``Learning from heterogeneous data based on social interactions over
  graphs,''
\newblock {\em IEEE Transactions on Information Theory}, vol. 69, no. 5, pp.
  3347--3371, 2023.

\bibitem{jadbabaie_2012}
A.~Jadbabaie, P.~Molavi, A.~Sandroni, and A.~Tahbaz-Salehi,
\newblock ``Non-{Bayesian} social learning,''
\newblock {\em Games and Economic Behavior}, vol. 76, no. 1, pp. 210--225,
  2012.

\bibitem{zhao_2012}
X.~Zhao and A.~H. Sayed,
\newblock ``Learning over social networks via diffusion adaptation,''
\newblock in {\em Proc. Asilomar Conference on Signals, Systems and Computers},
  2012, pp. 709--713.

\bibitem{nedic_2017}
A.~Nedić, A.~Olshevsky, and C.~A. Uribe,
\newblock ``Fast convergence rates for distributed non-{Bayesian} learning,''
\newblock {\em IEEE Transactions on Automatic Control}, vol. 62, no. 11, pp.
  5538--5553, 2017.

\bibitem{lalitha_2018}
A.~Lalitha, T.~Javidi, and A.~D. Sarwate,
\newblock ``Social learning and distributed hypothesis testing,''
\newblock {\em IEEE Transactions on Information Theory}, vol. 64, no. 9, pp.
  6161--6179, 2018.

\bibitem{parasnis2020non}
R.~Parasnis, M.~Franceschetti, and B.~Touri,
\newblock ``{Non-Bayesian} social learning on random digraphs with
  aperiodically varying network connectivity,''
\newblock {\em IEEE Transactions on Control of Network Systems}, vol. 9, no. 3,
  pp. 1202--1214, 2022.

\bibitem{acemoglu_2011}
D.~Acemoglu, M.~A. Dahleh, I.~Lobel, and A.~Ozdaglar,
\newblock ``Bayesian learning in social networks,''
\newblock {\em The Review of Economic Studies}, vol. 78, no. 4, pp. 1201--1236,
  2011.

\bibitem{hkazla2021bayesian}
J.~Hazla, A.~Jadbabaie, E.~Mossel, and M.~A. Rahimian,
\newblock ``Bayesian decision making in groups is hard,''
\newblock {\em Operations Research}, vol. 69, no. 2, pp. 632--654, 2021.

\bibitem{simon1990bounded}
H.~A. Simon,
\newblock ``Bounded rationality,''
\newblock in {\em Utility and Probability}, pp. 15--18. 1990.

\bibitem{conlisk1996bounded}
J.~Conlisk,
\newblock ``Why bounded rationality?,''
\newblock {\em Journal of Economic Literature}, vol. 34, no. 2, pp. 669--700,
  1996.

\bibitem{degroot1974}
M.~H. DeGroot,
\newblock ``Reaching a consensus,''
\newblock {\em Journal of the American Statistical Association}, vol. 69, no.
  345, pp. 118--121, 1974.

\bibitem{sayed_2014}
A.~H. Sayed,
\newblock ``{Adaptation, learning, and optimization over networks},''
\newblock {\em Foundations and Trends in Machine Learning}, vol. 7, no. 4-5,
  pp. 311--801, July 2014.

\bibitem{bordignon2020partial}
V.~Bordignon, V.~Matta, and A.~H. Sayed,
\newblock ``Partial information sharing over social learning networks,''
\newblock {\em IEEE Transactions on Information Theory}, vol. 69, no. 3, pp.
  2033--2058, 2023.

\bibitem{salhab2020}
R.~Salhab, A.~Ajorlou, and A.~Jadbabaie,
\newblock ``Social learning with sparse belief samples,''
\newblock in {\em Proc. IEEE CDC}, 2020, pp. 1792--1797.

\bibitem{mitraEvent}
A.~Mitra, S.~Bagchi, and S.~Sundaram,
\newblock ``Event-triggered distributed inference,''
\newblock in {\em Proc. IEEE CDC}, 2020, pp. 6228--6233.

\bibitem{paritoshAssignment}
P.~Paritosh, N.~Atanasov, and S.~Martinez,
\newblock ``Hypothesis assignment and partial likelihood averaging for
  cooperative estimation,''
\newblock in {\em Proc. IEEE CDC}, 2019, pp. 7850--7856.

\bibitem{inan2022social}
Y.~Inan, M.~Kayaalp, E.~Telatar, and A.~H. Sayed,
\newblock ``Social learning under randomized collaborations,''
\newblock in {\em Proc. IEEE ISIT}, 2022, pp. 115--120.

\bibitem{toghani2021communication}
M.~T. Toghani and C.~A. Uribe,
\newblock ``Communication-efficient distributed cooperative learning with
  compressed beliefs,''
\newblock {\em IEEE Transactions on Control of Network Systems}, vol. 9, no. 3,
  pp. 1215--1226, 2022.

\bibitem{mitra2021}
A.~Mitra, J.~A. Richards, S.~Bagchi, and S.~Sundaram,
\newblock ``Distributed inference with sparse and quantized communication,''
\newblock {\em IEEE Transactions on Signal Processing}, vol. 69, pp.
  3906--3921, 2021.

\bibitem{cirillo2022}
M.~Cirillo, V.~Bordignon, V.~Matta, and A.~H. Sayed,
\newblock ``Memory-aware social learning under partial information sharing,''
\newblock {\em IEEE Transactions on Signal Processing}, vol. 71, pp.
  2833--2848, 2023.

\bibitem{acemoglu13}
D.~Acemoglu, G.~Como, F.~Fagnani, and A.~Ozdaglar,
\newblock ``Opinion fluctuations and disagreement in social networks,''
\newblock {\em Mathematics of Operations Research}, vol. 38, no. 1, pp. 1--27,
  2013.

\bibitem{yildiz13}
E.~Yildiz, A.~Ozdaglar, D.~Acemoglu, A.~Saberi, and A.~Scaglione,
\newblock ``Binary opinion dynamics with stubborn agents,''
\newblock {\em ACM Trans. Econ. Comput.}, vol. 1, no. 4, pp. 1--30, 2013.

\bibitem{lena2019}
S.~D. Lena,
\newblock ``Non-{Bayesian} social learning and the spread of misinformation in
  networks,''
\newblock Working Papers 2019:09, Department of Economics, University of Venice
  Ca' Foscari, 2019,
\newblock available at SSRN: https://ssrn.com/abstract=3355245.

\bibitem{vial2021local}
D.~Vial and V.~Subramanian,
\newblock ``Local non-{Bayesian} social learning with stubborn agents,''
\newblock {\em IEEE Transactions on Control of Network Systems}, vol. 9, no. 3,
  pp. 1178--1188, 2022.

\bibitem{bordignon_2021}
V.~Bordignon, V.~Matta, and A.~H. Sayed,
\newblock ``Adaptive social learning,''
\newblock {\em IEEE Transactions on Information Theory}, vol. 67, no. 9, pp.
  6053--6081, 2021.

\bibitem{kayaalp_dist_bayesian}
M.~Kayaalp, V.~Bordignon, S.~Vlaski, V.~Matta, and A.~H. Sayed,
\newblock ``Distributed {Bayesian} learning of dynamic states,''
\newblock {\em arXiv:2212.02565}, Dec. 2022.

\bibitem{blondel09}
V.~D. Blondel, J.~M. Hendrickx, and J.~N. Tsitsiklis,
\newblock ``On {Krause's} multi-agent consensus model with state-dependent
  connectivity,''
\newblock {\em IEEE Transactions on Automatic Control}, vol. 54, no. 11, pp.
  2586--2597, 2009.

\bibitem{acemoglu2021}
D.~Acemoglu, A.~Ozdaglar, and J.~Siderius,
\newblock ``A model of online misinformation,''
\newblock Working Paper 28884, National Bureau of Economic Research, June 2021.

\bibitem{nedic2009distributed}
A.~Nedic and A.~Ozdaglar,
\newblock ``Distributed subgradient methods for multi-agent optimization,''
\newblock {\em IEEE Transactions on Automatic Control}, vol. 54, no. 1, pp.
  48--61, 2009.

\bibitem{sayed_proc2014}
A.~H. Sayed,
\newblock ``Adaptive networks,''
\newblock {\em Proc. IEEE}, vol. 102, no. 4, pp. 460--497, 2014.

\bibitem{dimakis2010_gossip}
A.~G. Dimakis, S.~Kar, J.~M.~F. Moura, M.~G. Rabbat, and A.~Scaglione,
\newblock ``Gossip algorithms for distributed signal processing,''
\newblock {\em Proc. IEEE}, vol. 98, no. 11, pp. 1847--1864, 2010.

\bibitem{sayed_2023}
A.~H. Sayed,
\newblock {\em Inference and Learning from Data},
\newblock Cambridge University Press, 2022,
\newblock 3 vols.

\bibitem{rabbat2014_asynchronous}
M.~G. Rabbat and K.~I. Tsianos,
\newblock ``Asynchronous decentralized optimization in heterogeneous systems,''
\newblock in {\em Proc. IEEE CDC}, 2014, pp. 1125--1130.

\bibitem{zhao_2015_p1}
X.~Zhao and A.~H. Sayed,
\newblock ``Asynchronous adaptation and learning over networks—part i:
  Modeling and stability analysis,''
\newblock {\em IEEE Transactions on Signal Processing}, vol. 63, no. 4, pp.
  811--826, 2015.

\bibitem{zhao_2015_p2}
X.~Zhao and A.~H. Sayed,
\newblock ``Asynchronous adaptation and learning over networks—part ii:
  Performance analysis,''
\newblock {\em IEEE Transactions on Signal Processing}, vol. 63, no. 4, pp.
  827--842, 2015.

\bibitem{zhao_2015_p3}
X.~Zhao and A.~H. Sayed,
\newblock ``Asynchronous adaptation and learning over networks—part iii:
  Comparison analysis,''
\newblock {\em IEEE Transactions on Signal Processing}, vol. 63, no. 4, pp.
  843--858, 2015.

\bibitem{lian_2018}
X.~Lian, W.~Zhang, C.~Zhang, and J.~Liu,
\newblock ``Asynchronous decentralized parallel stochastic gradient descent,''
\newblock in {\em Proc. International Conference on Machine Learning}, 10--15
  Jul 2018, vol.~80 of {\em Proceedings of Machine Learning Research}, pp.
  3043--3052.

\bibitem{bedi2019_asynchronous}
A.~S. Bedi, A.~Koppel, and K.~Rajawat,
\newblock ``Asynchronous saddle point algorithm for stochastic optimization in
  heterogeneous networks,''
\newblock {\em IEEE Transactions on Signal Processing}, vol. 67, no. 7, pp.
  1742--1757, 2019.

\bibitem{cao_2020}
X.~Cao and T.~Başar,
\newblock ``Decentralized multi-agent stochastic optimization with pairwise
  constraints and quantized communications,''
\newblock {\em IEEE Transactions on Signal Processing}, vol. 68, pp.
  3296--3311, 2020.

\bibitem{hanna2021_quantization}
O.~A. Hanna, Y.~H. Ezzeldin, C.~Fragouli, and S.~Diggavi,
\newblock ``Quantization of distributed data for learning,''
\newblock {\em IEEE Journal on Selected Areas in Information Theory}, vol. 2,
  no. 3, pp. 987--1001, 2021.

\bibitem{ghassemi_2015}
M.~Ghassemi and A.~D. Sarwate,
\newblock ``Distributed proportional stochastic coordinate descent with social
  sampling,''
\newblock in {\em Proc. Annual Allerton Conference on Communication, Control,
  and Computing}, 2015, pp. 17--24.

\bibitem{wang2017_coordinate}
C.~Wang, Y.~Zhang, B.~Ying, and A.~H. Sayed,
\newblock ``Coordinate-descent diffusion learning by networked agents,''
\newblock {\em IEEE Transactions on Signal Processing}, vol. 66, no. 2, pp.
  352--367, 2017.

\bibitem{hu2023_performance}
P.~Hu, V.~Bordignon, M.~Kayaalp, and A.~H. Sayed,
\newblock ``Performance of social machine learning under limited data,''
\newblock in {\em Proc. IEEE International Conference on Acoustics, Speech and
  Signal Processing (ICASSP)}, Rhodes Island, Greece, 2023, pp. 1--5.

\bibitem{yuan2018_mirror}
D.~Yuan, Y.~Hong, D.~W.~C. Ho, and G.~Jiang,
\newblock ``Optimal distributed stochastic mirror descent for strongly convex
  optimization,''
\newblock {\em Automatica}, vol. 90, pp. 196--203, 2018.

\bibitem{nedic2016}
A.~Nedić, A.~Olshevsky, and C.~A. Uribe,
\newblock ``A tutorial on distributed (non-{Bayesian}) learning: Problem,
  algorithms and results,''
\newblock in {\em Proc. IEEE CDC}, 2016, pp. 6795--6801.

\bibitem{koliander2022}
G.~Koliander, Y.~El-Laham, P.~M. Djurić, and F.~Hlawatsch,
\newblock ``Fusion of probability density functions,''
\newblock {\em Proc. IEEE}, vol. 110, no. 4, pp. 404--453, 2022.

\bibitem{li2019second}
T.~Li, H.~Fan, J.~Garc{\'\i}a, and J.~M. Corchado,
\newblock ``Second-order statistics analysis and comparison between arithmetic
  and geometric average fusion: Application to multi-sensor target tracking,''
\newblock {\em Information Fusion}, vol. 51, pp. 233--243, 2019.

\bibitem{kayaalp2022_aa_ga}
M.~Kayaalp, Y.~Inan, E.~Telatar, and A.~H. Sayed,
\newblock ``On the arithmetic and geometric fusion of beliefs for distributed
  inference,''
\newblock {\em IEEE Transactions on Automatic Control}, pp. 1--16, 2023.

\bibitem{uribe2022nonasymptotic}
C.~A. Uribe, A.~Olshevsky, and A.~Nedi{\'c},
\newblock ``Nonasymptotic concentration rates in cooperative learning—part
  ii: Inference on compact hypothesis sets,''
\newblock {\em IEEE Transactions on Control of Network Systems}, vol. 9, no. 3,
  pp. 1141--1153, 2022.

\bibitem{kar2012distributed}
S.~Kar, J.~M.~F. Moura, and K.~Ramanan,
\newblock ``Distributed parameter estimation in sensor networks: Nonlinear
  observation models and imperfect communication,''
\newblock {\em IEEE Transactions on Information Theory}, vol. 58, no. 6, pp.
  3575--3605, 2012.

\bibitem{dedecius_2017}
K.~Dedecius and P.~M. Djurić,
\newblock ``Sequential estimation and diffusion of information over networks: A
  {Bayesian} approach with exponential family of distributions,''
\newblock {\em IEEE Transactions on Signal Processing}, vol. 65, no. 7, pp.
  1795--1809, 2017.

\bibitem{horn2012matrix}
R.~A. Horn and C.~R. Johnson,
\newblock {\em Matrix Analysis},
\newblock Cambridge University Press, 2012.

\bibitem{williams_1991}
D.~Williams,
\newblock {\em Probability with Martingales},
\newblock Cambridge University Press, 1991.

\bibitem{cesaro1888convergence}
E.~Ces{\`a}ro,
\newblock ``Sur la convergence des s{\'e}ries,''
\newblock {\em Nouvelles Annales de Math{\'e}matiques: Journal des Candidats
  aux {\'E}coles Polytechnique et Normale}, vol. 7, pp. 49--59, 1888.

\end{thebibliography}

\end{document}